\documentclass[letterpaper,11pt]{article}

\usepackage{times}
\usepackage{fullpage}
\usepackage{microtype}
\usepackage{amsmath,amsfonts,amssymb}
\usepackage[amsmath,amsthm,thmmarks,hyperref]{ntheorem}
\usepackage{mathtools}
\usepackage{url}
\usepackage{enumitem}
\usepackage{xspace}
\usepackage{hyperref}
\usepackage{algorithm}
\usepackage[noend]{algorithmic}

\newcommand{\N}{\ensuremath{\mathbb{N}}}

\newcommand{\R}{\ensuremath{\mathbb{R}}}

\newcommand{\Z}{\ensuremath{\mathbb{Z}}}

\newcommand{\pr}[2]{\left\langle#1, #2\right\rangle}
\newcommand{\inner}[1]{\langle{#1}\rangle}

\newcommand{\set}[1]{\{{#1}\}}

\newcommand{\round}[1]{\lfloor{#1}\rceil}

\newcommand{\length}[1]{\lVert{#1}\rVert}

\newcommand{\veca}{\ensuremath{\mathbf{a}}}
\newcommand{\vecb}{\ensuremath{\mathbf{b}}}

\newcommand{\vece}{\ensuremath{\mathbf{e}}}

\newcommand{\vecs}{\ensuremath{\mathbf{s}}}
\newcommand{\vect}{\ensuremath{\mathbf{t}}}
\newcommand{\vecu}{\ensuremath{\mathbf{u}}}
\newcommand{\vecv}{\ensuremath{\mathbf{v}}}
\newcommand{\vecw}{\ensuremath{\mathbf{w}}}
\newcommand{\vecx}{\ensuremath{\mathbf{x}}}
\newcommand{\vecy}{\ensuremath{\mathbf{y}}}
\newcommand{\vecz}{\ensuremath{\mathbf{z}}}
\newcommand{\veczero}{\ensuremath{\mathbf{0}}}

\theoremstyle{plain}            %
\newtheorem{theorem}{Theorem}[section]
\newtheorem{lemma}[theorem]{Lemma}
\newtheorem{corollary}[theorem]{Corollary}
\newtheorem{proposition}[theorem]{Proposition}
\newtheorem{claim}[theorem]{Claim}

\theoremstyle{definition}       %
\newtheorem{definition}[theorem]{Definition}

\theoremstyle{remark}           %
\newtheorem{remark}[theorem]{Remark}

\numberwithin{equation}{section}
\DeclareMathOperator{\poly}{poly}
\DeclareMathOperator{\polylog}{polylog}
\DeclareMathOperator{\negl}{negl}

\DeclareMathOperator*{\E}{\mathbb{E}}

\newcommand{\problem}[1]{\ensuremath{\mathsf{#1}}\xspace}

\newcommand{\class}[1]{\ensuremath{\mathsf{#1}}\xspace}

\newcommand{\NP}{\class{NP}}
\newcommand{\coNP}{\class{coNP}}
\newcommand{\AM}{\class{AM}}
\newcommand{\coAM}{\class{coAM}}

\newcommand{\SZK}{\class{SZK}}

\newcommand{\HVSZK}{\class{HVSZK}}

\newcommand{\lat}{\mathcal{L}}
\DeclareMathOperator{\vol}{vol}
\DeclareMathOperator{\dist}{dist}

\newcommand{\gapsvp}{\problem{GapSVP}}

\newcommand{\gapcvp}{\problem{GapCVP}}

\newcommand{\gapsp}{\problem{GapSPP}}
\newcommand{\gapspp}{\problem{GapSPP}}
\newcommand{\cogapsp}{\problem{coGapSP}}

\newif\ifnotes\notesfalse

\ifnotes
\usepackage{color}
\definecolor{mygrey}{gray}{0.50}
\newcommand{\notename}[2]{{\textcolor{mygrey}{\footnotesize{\bf (#1:} {#2}{\bf ) }}}}

\else

\newcommand{\notename}[2]{{}}

\fi

\newcommand{\cnote}[1]{{\notename{Chris}{#1}}}
\newcommand{\dnote}[1]{{\notename{Daniel}{#1}}}
\newcommand{\knote}[1]{{\notename{Kai-Min}{#1}}}
\newcommand{\fnote}[1]{{\notename{Feng-Hao}{#1}}}

\newcommand{\ID}{{\text{ID}}}      %
\newcommand{\Com}{{\text{Com}}}      %
\newcommand{\SPCom}{{\text{SPCom}}}      %

\newcommand{\PPT}{{\text{PPT}}}      %
\newcommand{\supp}{{\text{supp}}}      %

\newcommand{\zo}{{\{0,1\}}}      %
\newcommand{\LWE}{{\sf LWE}{}}     
\newcommand{\BDD}{{\sf BDD}{}}

\newcommand{\Prov}{{\text{P}}}      %
\newcommand{\Ver}{{\text{V}}}      %
\newcommand{\eqdef}{\mathbin{\stackrel{\rm def}{=}}}

\renewcommand{\epsilon}{\varepsilon}
\newcommand{\eps}{\epsilon}
\newcommand{\calV}{\mathcal{V}}

\title{On the Lattice Smoothing Parameter Problem}

\author{%
  Kai-Min Chung \\ Cornell University
  \and
  Daniel Dadush \\ New York University
  \and
  Feng-Hao Liu \\ Brown University
  \and
  Chris Peikert \\ Georgia  Institute of Technology
}

\begin{document}

\maketitle
\thispagestyle{empty}

\begin{abstract}
The \emph{smoothing parameter} $\eta_{\epsilon}(\lat)$ of a Euclidean
lattice $\lat$, introduced by Micciancio and Regev (FOCS'04;
SICOMP'07), is (informally) the smallest amount of Gaussian noise that
``smooths out'' the discrete structure of $\lat$ (up to error
$\epsilon$).  It plays a central role in the best known
worst-case/average-case reductions for lattice problems, a wealth of
lattice-based cryptographic constructions, and (implicitly) the
tightest known transference theorems for fundamental lattice
quantities.

In this work we initiate a study of the complexity of approximating
the smoothing parameter to within a factor $\gamma$, denoted
$\gamma$-$\gapsp$.  We show that (for $\eps = 1/\poly(n)$):
\knote{or for $1/\poly(n) \leq \eps \leq o(1)$?}
\begin{itemize}
\item $(2+o(1))$-$\gapsp \in \AM$, via a Gaussian analogue of the
  classic Goldreich-Goldwasser protocol (STOC'98);
\item $(1+o(1))$-$\gapsp \in \coAM$, via a careful application of the
  Goldwasser-Sipser (STOC'86) set size lower bound protocol to thin
  shells in $\R^{n}$;
\item $(2+o(1))$-$\gapsp \in \SZK \subseteq \AM \cap \coAM$
  (where $\SZK$ is the class of problems having statistical
  zero-knowledge proofs), by constructing a suitable
  instance-dependent commitment scheme (for a slightly worse $o(1)$-term);
\item $(1+o(1))$-$\gapsp$ can be solved in deterministic $2^{O(n)}
  \polylog(1/\epsilon)$ time and $2^{O(n)}$ space.
\end{itemize}
As an application, we demonstrate a tighter worst-case to average-case
reduction for basing cryptography on the worst-case hardness of the $\gapspp$
problem, with $\tilde{O}(\sqrt{n})$ smaller approximation factor than the
$\gapsvp$ problem. Central to our results are two novel, and nearly tight,
characterizations of the magnitude of discrete Gaussian sums over $\lat$: the
first relates these directly to the Gaussian measure of the Voronoi cell of
$\lat$, and the second to the fraction of overlap between Euclidean balls
centered around points of $\lat$.

  \ifnotes \begin{center}{\Huge{NOTES ARE ON }}\end{center} \fi
\end{abstract}

\newpage
\pagenumbering{arabic}

\section{Introduction}
\label{sec:introduction}

A (full-rank) $n$-dimensional lattice $\lat=\lat(B)=\set{\sum_{i=1}^{n} c_{i} \vecb_{i} : c_{i} \in \Z}$ is the set of all integer linear combinations
of a set $B=\set{\vecb_{1}, \ldots, \vecb_{n}} \subset \R^{n}$ of linearly independent vectors, called a basis of the lattice.  It may also be seen as
a discrete additive subgroup of $\R^{n}$.  Lattices have been studied in mathematics for hundreds of years, and more recently have been at the center
of many important developments in computer science, such as the LLL algorithm~\cite{lenstra82:_factor} and its applications to
cryptanalysis~\cite{DBLP:journals/joc/Coppersmith97} and error-correcting codes~\cite{DBLP:conf/innovations/CohnH11}, and lattice-based
cryptography~\cite{ajtai04:_gener_hard_instan_lattic_probl} (including the first fully homomorphic encryption scheme~\cite{DBLP:conf/stoc/Gentry09}).

Much recent progress in the computational study of lattices, especially in the realms of worst-case/average-case reductions and cryptography (as
initiated by Ajtai~\cite{ajtai04:_gener_hard_instan_lattic_probl}), has been made possible by the machinery of Gaussian measures and harmonic
analysis.  These tools were first employed for such purposes by Regev~\cite{DBLP:journals/jacm/Regev04} and Micciancio and
Regev~\cite{DBLP:journals/siamcomp/MicciancioR07} (see also,
e.g.,~\cite{DBLP:journals/jacm/AharonovR05,DBLP:journals/jacm/Regev09,peikert08:_limit_hardn_of_lattic_probl,DBLP:conf/stoc/GentryPV08,DBLP:conf/crypto/Gentry10,DBLP:conf/crypto/Peikert10}),
following their development by Banaszczyk~\cite{banaszczyk93:_new,DBLP:journals/dcg/Banaszczyk95,Bana96} to prove asymptotically tight (or nearly
tight) transference theorems.

In particular, the notion
from~\cite{DBLP:journals/siamcomp/MicciancioR07} of the
\emph{smoothing parameter} $\eta_{\eps}(\lat)$ of a lattice~$\lat$
plays a central role (sometimes implicitly) the above-cited works, and
so it is a key concept in the study of lattices from several
perspectives.  Informally, $\eta_{\epsilon}(\lat)$ is the smallest
amount $s$ of Gaussian noise that completely ``smooths out'' the
discrete structure of $\lat$, up to statistical error $\epsilon$.
Formally, it is the smallest $s>0$ such that the total Gaussian mass
$\rho_{1/s}(\vecw) := \exp(-\pi s^{2} \length{\vecw}^{2})$, summed
over all nonzero \emph{dual} lattice vectors $\vecw \in \lat^{*}
\setminus \set{\veczero}$, is at most~$\epsilon$.\footnote{The dual
  lattice $\lat^{*}$ of $\lat$ is the set of all $\vecy \in \R^{n}$
  for which $\inner{\vecx,\vecy} \in \Z$ for every $\vecx \in \lat$.}
This condition is equivalent to the following ``smoothing'' condition:
the distribution of a continuous Gaussian of width $s$, reduced modulo
$\lat$, has point-wise probability density within a $(1\pm \eps)$
factor of that of the uniform distribution over $\mathbb{R}^{n}/\lat$.
\cnote{I don't think this very long and technical description is
  helpful at this point; let's just go straight to the definition
  instead...}
\knote{I agree that formal technical description is too much here, but I feel that giving a bit more interpretation is helpful to bridge the informal ``smooth out'' and the formal Gaussian sum. I try to give a soft description here; perhaps we can further soften it?}
\cnote{How about the}

Given the smoothing parameter's central role in many mathematical and
computational aspects of lattices, we believe it to be of comparable
importance to other fundamental and well-studied geometric lattice
quantities like the minimum distance, successive minima, covering
radius, etc.  While the smoothing parameter can be estimated by
relating it to these other
quantities~\cite{DBLP:journals/siamcomp/MicciancioR07,peikert08:_limit_hardn_of_lattic_probl,DBLP:conf/stoc/GentryPV08},
the bounds are quite coarse, typically yielding only
$\tilde{\Omega}(\sqrt{n})$-factor approximations.

We therefore initiate a study of the complexity of computing the
smoothing parameter, with a focus on approximations. More formally,
for an approximation factor $\gamma \geq 1$ and some $0 < \eps < 1$
(which may both be functions of the lattice dimension~$n$), we define
$\gamma$-$\gapsp_\eps$ to be the promise problem in which YES
instances are lattices $\lat$ for which $\eta_{\eps}(\lat) \leq 1$,
and NO instances are those for which $\eta_{\eps}(\lat) > \gamma$.

\paragraph{The dependence on $\eps$.}

To understand the nature of $\gapspp$, it is important to notice that
the value of $\eps$ has a large impact on the complexity of the
problem. In particular, by known relations between the smoothing
parameter and the shortest nonzero dual vector
(see~\cite{DBLP:journals/siamcomp/MicciancioR07}), we have that
\[ \sqrt{\log (1/\eps)/\pi} / \lambda_1(\lat^*) \leq \eta_{\eps}(\lat)
\leq \sqrt{n}/\lambda_1(\lat^*), \] and hence for exponentially small
error $\eps = 2^{-\Omega(n)}$ the quantities $\eta_{\eps}(\lat)$ and
$\sqrt{n}/\lambda_1(\lat^*)$ are within a constant factor of each
other. Therefore, the (decision) Shortest Vector Problem
$\gamma$-$\gapsvp$ is equivalent to
$\gamma$-$\gapspp_{2^{-\Omega(n)}}$, up to a constant factor loss in
the approximation. However, most uses of the smoothing parameter in
the literature (e.g., worst-case to average-case reductions and
transference theorems) work with either inverse polynomial $\eps =
n^{-O(1)}$ or ``just barely'' negligible $\eps = n^{-\omega(1)}$
(e.g., $\epsilon = n^{-\log n}$).  For such values of $\epsilon$, the
loss in approximation factor between $\gapspp$ and $\gapsvp$ or other
standard lattice problems can be as large as
$\tilde{\Omega}(\sqrt{n})$, and as we will see, in this regime
$\gapspp$ behaves qualitatively differently from these other problems.

\subsection{Results and Techniques}

In this work, we prove several (possibly surprising) upper bounds on
the complexity of $\gamma$-$\gapsp_\eps$. Unless otherwise specified,
the stated results hold for the setting $\eps = n^{-O(1)}$. (We obtain
results for smaller $\eps$ as well, but with slowly degrading
approximation factors.)  Similar results hold for a generalization of
$\gapspp$ which uses different values of $\eps$ for YES and NO
instances (see Definition~\ref{def:gapsp} and
Corollary~\ref{cor:eps-epsY-epsN} for further details).

At a high level, we obtain several of our main results by noticing
that the classic Goldreich-Goldwasser
protocol~\cite{DBLP:journals/jcss/GoldreichG00}, which was originally
designed for approximating (the complement of) the $\gapsvp$ problem,
can in fact be seen as more directly and tightly approximating the
smoothing parameter (of the dual lattice). When viewed from this
perspective, we show that slight variants of the GG protocol obtain an
$2+o(1)$ approximation for $\gapspp$, improving on the approximation
for $\gapsvp$ by a $\tilde{O}(\sqrt{n})$ factor. Furthermore, using
the known relations between $\gapsvp$ and $\gapspp$, one recover the
original approximation factor for $\gapsvp$. 
To obtain these tight approximation factors, as part of the
main technical contributions of this paper, we develop two novel and
nearly tight (up to a $2+o(1)$ factor) geometric characterizations of
the smoothing parameter $\eta_\eps(\lat)$ that elucidate the geometric
content of the parameter $\eps$.

\paragraph{Arthur-Merlin Protocols.}

We show that $(2+o(1))$-$\gapsp \in \AM ~\cap~ \coAM$, and
moreover, that $(1+o(1))$-$\gapsp \in \coAM$.  That is, we give
constant-round interactive proof systems which allow an unbounded
prover to convince a randomized polynomial-time verifier that the
smoothing parameter is ``small,'' and that it is ``large.'' In
contrast with these positive results, we note that since the smoothing
parameter is effectively determined by a sum over exponentially many
lattice points, it is unclear whether $\gamma$-$\gapsp$ is in $\NP$ or
$\coNP$ for $\gamma=o(\sqrt{n})$.  (For $\gamma=\Omega(\sqrt{n})$,
known connections to other lattice quantities imply that
$\gamma$-$\gapspp_{\epsilon} \in\NP \cap \coNP$.)

One important consequence of $(2+o(1))$-$\gapsp \in \AM \cap \coAM$ is
that the problem is not $\NP$-hard (under Karp reductions, or
``smart'' Cook reductions~\cite{DBLP:journals/siamcomp/GrollmannS88}),
unless $\coNP \subseteq \AM$~\cite{DBLP:journals/ipl/BoppanaHZ87} and
the polynomial-time hierarchy collapses.  Our result should also be
contrasted with analogous results for approximating the Shortest and
Closest Vector Problems, which are only known to be in $\NP \cap
\coAM$ for factors $\gamma \geq c\sqrt{n/\log
  n}$~\cite{DBLP:journals/jcss/GoldreichG00}, and in $\NP \cap \coNP$
for factors $\gamma \geq
c\sqrt{n}$~\cite{DBLP:journals/jacm/AharonovR05}, as well as the
results for approximating the Covering Radius Problem, whose
$2$-approximation is in $\AM$ but is in $\coAM$ only for $\gamma \geq
c\sqrt{n/\log n}$, and in $\NP \cap \coNP$ for $\gamma \geq \sqrt{n}$
~\cite{DBLP:journals/cc/GuruswamiMR05}.

To prove that $(2+o(1))$-$\gapsp \in \AM$, we use a Gaussian analogue
of the Goldreich-Goldwasser protocol on the dual lattice $\lat^*$,
where the verifier samples from a Gaussian instead of a
ball. (Interestingly, this leads to imperfect completeness, which
turns out to be important for the tightness of the analysis.)  More
precisely, the verifier samples $\vecx \in \R^{n}$ from a Gaussian,
reduces $\vecx$ modulo (a basis of) the lattice $\lat^*$, and sends
the result to the prover. The prover's task is to guess $\vecx$, and
the verifier accepts or rejects accordingly. To prove that the
protocol is complete and sound, we crucially rely on the following
novel characterization of the smoothing parameter:
\begin{itemize}
\item[] {\bf Voronoi Cell Characterization.} For any $\eps \in (0,1)$,
  a scaling of the Voronoi cell\footnote{The Voronoi cell
    $\calV(\lat^{*})$ is the set of points in $\R^{n}$ that are closer
    to $\veczero$ than any other lattice point of $\lat^*$, under
    $\ell_2$ norm.} $\calV(\lat^*)$ by a factor $2\eta_\eps(\lat)$ has
  Gaussian measure at least $1-\eps$, and an $\eta_\eps(\lat)$-scaling
  has Gaussian measure at most $1/(1+\eps)$.
\end{itemize}
With this tool in hand, the analysis of the protocol is very
simple. By the maximum likelihood principle, the optimal prover
guesses correctly if and only if the verifier's original sample lands
inside the Voronoi cell, and hence the verifier's acceptance probability
is exactly the Gaussian measure of $\calV(\lat^*)$. See
Section~\ref{sec:AM} for further details.

For proving $(1+o(1))$-$\gapsp \in \coAM$, we rely on the classic
set-size lower bound protocol of Goldwasser and
Sipser~\cite{GoldwasserS86}.  In order to prove that the discrete
Gaussian mass on $\lat^{*} \setminus \set{\veczero}$ is large, we
apply the protocol to thin shells in $\R^{n}$, and rely on a discrete
Gaussian concentration inequality of
Banaszczyk~\cite{banaszczyk93:_new}. See Section~\ref{sec:coAM} for an
overview and full details.

\paragraph{\bf Statistical Zero Knowledge Protocol.}

We prove that $(2+o(1))$-$\gapsp \in \SZK$, the class of problems
having statistical zero-knowledge proofs. We note that this result
does not subsume the inclusion in $\AM \cap \coAM$ described above (as
one might suspect, given that $\SZK \subseteq \AM \cap \coAM$), due to
a slightly worse dependence $\eps$ in the~$o(1)$ term. To prove the
theorem, we construct a new instance-dependent commitment
scheme\footnote{Roughly speaking, an instance-dependent commitment
  scheme for a language $L$ is a commitment scheme that can depend on
  the instance $x$ and such that only one of the (statistical) hiding
  and binding properties are required to hold, depending on whether $x
  \in L$.} based on $\gapspp$, which is sufficiently binding (for an
honest committer) and hiding (to a dishonest receiver). Constructing
such a commitment scheme (with some additionals observations in our
case) is known to be sufficient for obtaining an $\SZK$
protocol~\cite{ItohOS97}.

Our construction can be viewed as a generalization of an
instance-dependent commitment scheme for $O(\sqrt{n/\log
  n})$-$\gapsvp$ implicit in~\cite{DBLP:conf/crypto/MicciancioV03},
which was also based on the Goldreich-Goldwasser protocol and is
perfectly binding. At a very high level, the commitment scheme is
based on revealing a ``random'' perturbed lattice point in $\lat$,
where the perturbation is taken uniformly from a ball of radius
$r$. Roughly speaking, we get the binding property when there is only
one lattice within distance $r$ of the revealed perturbation, and get
the hiding property when there are multiple such lattice points (which
allow for equivocation). It turns out that the main measure of quality
for the binding and hiding property corresponds to the fraction of
overlap between the balls of radius $r$ placed around lattice points
of $\lat$: less overlap means better binding, and more overlap yields
better hiding. In \cite{DBLP:conf/crypto/MicciancioV03}, this overlap
is analyzed in terms of the length $\lambda_{1}$ of the shortest
nonzero vector of~$\lat$. In particular, if $r \leq \lambda_1/2$, then
the balls are completely disjoint (perfect binding), and if $r \geq
\Omega(\sqrt{n/\log n}) \cdot \lambda_{1}$, then a $1/\poly(n)$
fraction of the ball around any lattice point overlaps with that of
its nearest neighbor in the lattice, which gives non-negligible
hiding.

The main insight which allows us obtain improved approximation factors
when basing the commitment scheme on $\gapspp$ is a new
characterization of the smoothing parameter, which allows to get very
fine control on the overlap.

\begin{itemize}
\item[] {\bf Ball Overlap Characterization}. For $\eps \geq
  2^{-o(n)}$, Euclidean balls of radius $R=\sqrt{n/(2\pi)} /
  (2\eta_\eps(\lat^*))$ centered at all points of $\lat$ overlap in at
  most a $2\eps$ fraction of their mass, and balls of radius
  $(2+o(1))R$ overlap in at least an $\eps/2$ fraction of their mass.
\end{itemize}

From the above we are able to determine, to within a factor $2+o(1)$,
whether balls placed at points of~$\lat$ overlap in at most or at
least an $\eps$ fraction of their mass, based solely on the smoothing
parameter (of the dual lattice). Intuitively, this is because the
smoothing parameter takes into account all the lattice points
in~$\lat$, and hence is able to provide much better ``global''
information about the overlap. We refer the reader to
Section~\ref{sec:SZK} for further details and discussion.

\paragraph{Application to Worst-Case/Average-Case Reductions.}

As an application, we also obtain a worst-case to average-case
reduction from $\gapspp$ to the Learning With Errors problem
($\LWE$)~\cite{DBLP:journals/jacm/Regev09}, which has a tighter
connection factor than the known reductions from
$\gapsvp$~\cite{DBLP:journals/jacm/Regev09,DBLP:conf/stoc/Peikert09}. Roughly
speaking, the goal of $\LWE$ is to solve $n$-dimensional random noisy
linear equations modulo some $q$, where Gaussian noise with standard
deviation $\alpha q$ is added to each equation. The $\LWE$ problem is
extremely versatile as a basis for numerous cryptographic
constructions
(e.g.,~\cite{DBLP:conf/stoc/PeikertW08,DBLP:conf/stoc/GentryPV08,DBLP:conf/eurocrypt/CashHKP10,DBLP:conf/focs/BrakerskiV11}).
Regev's celebrated result~\cite{DBLP:journals/jacm/Regev09} showed a
quantum reduction from solving worst-case $\gamma$-$\gapsvp$ (among
other problems) to solving $\LWE$ with $\gamma =
\tilde{O}(n/\alpha)$. Furthermore,
Peikert~\cite{DBLP:conf/stoc/Peikert09} showed a corresponding
classical reduction, when the modulus $q\geq 2^{n/2}$. Therefore, the
security of $\LWE$-based cryptographic constructions can be based on
the worst-case hardness of the $\gapsvp$ problem.

We observe that the reductions
of~\cite{DBLP:journals/jacm/Regev09,DBLP:conf/stoc/Peikert09} in fact
implicitly solve the $\gapspp$ problem. Thus, by slightly modifying
the last step of those reductions, we obtain corresponding
quantum/classical reductions from $\gamma$-$\gapspp_{\epsilon}$ (with
$\epsilon = \negl(n)$) to $\LWE$ with $\gamma = O(\sqrt{n}/\alpha)$.
As a consequence, the security of $\LWE$-based cryptographic
constructions can be based on the worst-case hardness of a potentially
harder lattice problem.

The application to worst-case/average-case reduction follows by noting
that the reduction of~\cite{DBLP:conf/stoc/Peikert09} solves $\gapsvp$
by running the Goldreich-Goldwasser protocol, where the prover's
strategy is simulated by using a bounded distance decoding (\BDD)
oracle, which in turn is implemented using the $\LWE$ oracle. To
obtain a tighter reduction from $\gapspp$ to $\LWE$, we observe that
the quality of the $\BDD$ oracle depends directly on the smoothing
parameter, as opposed to the length of the shortest vector. In light
of this, we instead solve $\gapspp$ using the Gaussian analogue of the
Goldreich-Goldwasser protocol described above, while still using a bounded
distance decoding (\BDD) oracle to simulate the prover's strategy. See
Section~\ref{sec:wstavg} for further details.

\paragraph{Algorithm for $\gapspp$.}

We give a deterministic $2^{O(n)} \polylog(1/\epsilon)$-time and
$2^{O(n)}$-space algorithm for deciding $(1+o(1))$-$\gapsp$.  For this
we use recent algorithms
of~\cite{DBLP:conf/stoc/MicciancioV10,DBLP:conf/focs/DadushPV11} for
enumerating lattice points in $\lat^{*}$ to estimate the Gaussian
mass.  The full details are in Section~\ref{col:det}.

\paragraph{Perspectives and Open Questions.} 

Our initial work on the complexity of the $\gapspp$ problem opens up
several directions for further study of the smoothing parameter from a
computational perspective. Perhaps the most intriguing question is
whether $(2+o(1))$-$\gapspp$ is $\SZK$-complete.  A positive answer
might lead to progress on the long-standing goal of basing
cryptography on general complexity classes.  Some reason for optimism
comes from its rather unusual complexity: like $\SZK$-complete
problems, $(2+o(1))$-$\gapspp$ is in $\SZK$ but is not known to be in
$\NP$ or $\coNP$.  We are unaware of any other problems (aside from
$\SZK$-complete ones) having these characteristics.

In a related direction, in this work we focus on the standard
``$L_{\infty}$ notion'' of the smoothing
parameter~$\eta_{\eps}(\lat)$, whereas the complexity of a related
``$L_1$ notion'' of the smoothing parameter, denoted
$\eta^{(1)}_{\eps}(\lat)$, also seems quite interesting. More
precisely, $\eta_{\eps}(\lat)$ can be defined equivalently as the
smallest parameter~$s$ such that the distribution of a continuous
Gaussian of width $s$, reduced modulo $\lat$, has point-wise
probability density within a $(1\pm \eps)$ factor of that of the
uniform distribution on $\R^{n}/\lat$. The $L_1$ variant
$\eta^{(1)}_\eps(\lat)$ of the smoothing parameter instead is defined
to be the smallest parameter $s$ such that the statistical distance
(i.e., half of the $L_1$ distance) between the above two distributions
is at most $\eps$.  (Clearly, $\eta^{(1)}_{\epsilon}(\lat) \leq
\eta_{\epsilon}(\lat)$.)  By definition, the problem of approximating
the $L_1$ smoothing parameter, denoted $\gamma$-$\gapspp^{(1)}_\eps$,
appears to naturally reduce to a well-known $\SZK$-complete problem
called Statistical Difference (SD) problem~\cite{SahaiV03}, which is a
promise problem asking whether two input distributions (specified by
circuits) have statistical distance less than $\alpha$ or greater than
$\beta$. Thus, the problem appears to be in $\SZK$ and is another
candidate $\SZK$-complete lattice problem. Unfortunately, the above
argument relies on $\eta^{(1)}_\eps(\lat)$ being a monotonic function
in~$\eps$, which is a basic property that we do not know how to prove
(or disprove)! In fact, we know very little about the $L_1$ smoothing
parameter. Given the potentially interesting complexity of
$\gamma$-$\gapspp^{(1)}_\eps$, it seems worthwhile to further
investigate the $L_1$ smoothing parameter, from both the geometric and
computational perspectives.
  
Finally, we note that our results generally apply only in the setting
where $\eps < 1$. It seems quite interesting to understand how the
complexity of $\gapspp$ changes for larger $\eps$. We remark that our
geometric characterizations only ``half fail'' for larger $\eps$. More
precisely, in the regime \cnote{Do you need to say
  $\eta_{\epsilon}(\lat) \geq 1$ here?:} $\eta_\eps(\lat) \geq 1$,
$\eps \geq 1$, we still get upper bounds on the Gaussian measure of
the Voronoi cell, as well as lower bounds on the fraction of overlap
for balls centered at lattice points. For our $\AM$ protocol, this
implies that the prover generally fails to convince the verifier, and
for our instant-dependent commitment scheme, this implies that it is
always hiding. \cnote{What do these next two sentences
  mean?:}Interestingly, our $\coAM$ protocol still applies for larger
$\eps$, almost without change. Here the main issue is that we do not
know a ``good'' geometric interpretation of the statement $\rho(\lat
\setminus \set{\veczero}) \leq \eps$ for any $\eps \geq 1$.

\paragraph{Organization.}

The rest of the paper is organized as follows.  In
Section~\ref{sec:preliminaries} we give the basic preliminaries.  In
Section~\ref{sec:AM}, we give our Arthur-Merlin protocol for showing
that $(2+o(1))$-$\gapsp \in \AM$ (Theorem~\ref{thm:GGG}). In
Section~\ref{sec:SZK} we construct a statistical zero-knowledge proof
for $\gapsp$ (Theorem \ref{thm:szk}). In Section~\ref{sec:wstavg}, we
describe the reduction from $\gapspp$ to $\LWE$
(Theorem~\ref{thm:gapsp-to-lwe}).  In Section~\ref{sec:coAM}, we show
that $(1+o(1))$-$\gapsp \in \coAM$ (Theorem~\ref{thm:coam}). In
Section~\ref{col:det} we give a deterministic algorithm for computing
the smoothing parameter (Theorem~\ref{thm:det}).

\section{Preliminaries}
\label{sec:preliminaries}

For sets $A,B \subseteq \R^n$, denote their Minkowski sum by $A+B =
\set{\veca+\vecb: \veca \in A, \vecb \in B}$.  We let $B_2^n = \set{\vecx \in \R^n:
  \length{\vecx}_2 \leq 1}$ denote the unit Euclidean ball in $\R^n$, and
$S^{n-1} = \partial B_2^n$ the unit sphere in $\R^n$.  Unless
stated otherwise, $\length{\cdot}$ denotes the Euclidean norm.

\paragraph{Lattices.}

A lattice $\lat \subset \R^{n}$ with basis $B$, and its dual
$\lat^{*}$, are defined as in the introduction.  For a basis $B$ and a
vector $\vecx \in \R^{n}$, we let $\vecx \bmod B$ denote the unique
$\bar{\vecx} \in \lat+\vecx$ such that $\bar{\vecx} = \sum_{i=1}^{n}
c_{i} \vecb_{i}$ for $c_{i} \in [-\tfrac12,\tfrac12)$.  It can be
computed efficiently from $\vecx$ and $B$ (treated as matrix of column
vectors) as $\bar{\vecx}=\vecx-B\round{B^{-1} \vecx}$.  We sometimes
instead write $\vecx \bmod \lat$ when the basis is implicit.

The Voronoi cell $\calV(\lat)$ is the set of points in $\R^n$ that are
at least as close to $\veczero$ (under the $\ell_{2}$ norm) as to any
other vector in $\lat$:
\begin{align*}
  \calV(\lat) &= \set{\vecx \in \R^n: \length{\vecx}_2 \leq
    \length{\vecx-\vecy}_2,\; \forall\; \vecy \in \lat \setminus
    \set{\veczero}} \\
  &= \set{\vecx \in \R^n: \pr{\vecx}{\vecy} \leq
    \tfrac{1}{2} \pr{\vecy}{\vecx},\; \forall\; \vecy \in \lat \setminus
    \set{\veczero}}.
\end{align*}
When the lattice in question is clear we shorten $\calV(\lat)$ to
$\calV$.  Note that $\calV$ is a symmetric polytope that tiles space
with respect to $\lat$, i.e., $\lat + \calV = \R^n$ and for all
distinct $x,y \in \lat$, the sets $x + \calV$ and $y + \calV$ are
interior disjoint.

\paragraph{Gaussian measures.}

Define the Gaussian function $\rho \colon \R^{n} \to \R^{+}$ as
$\rho(\vecx)=e^{-\pi \length{\vecx}^{2}}$, and for real $s>0$, define
$\rho_{s}(\vecx)=\rho(\vecx/s)=e^{-\pi
  \length{\vecx}^{2}/s^{2}}$.  For a countable subset $A \subseteq \R^n$, we
define $\rho_{s}(A) = \sum_{\vecx \in A} \rho_{s}(\vecx)$.

For a measurable subset $A \subseteq \R^n$, we define the Gaussian
measure of $A$ (parameterized by $s>0$) as $\gamma_s(A) =
\frac{1}{s^n} \int_A \rho_s(\vecx)\, d\vecx$.  Note that
$\gamma_s(\R^n) = 1$, so $\gamma_{s}$ is a probability measure.  For
parameter $s > 0$, we let $D_{s}$ be the corresponding continuous Gaussian
distribution with parameter $s$ centered around $\veczero$:
\[
D_{s}(A) = \gamma_s(A) \quad \forall \text{ measurable } A \subseteq \R^n.
\]
Similarly, for any countable subset $T \subseteq \R^n$ for which
$\rho_{s}(T)$ converges, define the discrete Gaussian
distribution $D_{T,s}$ over $T$ by
\[
D_{T,s}(\vecx) = \frac{\rho_{s}(\vecx)}{\rho_{s}(T)}
\quad \forall\; \vecx \in T.
\]
We usually consider the discrete Gaussian over a lattice $\lat$, i.e.,
where $T = \lat$, though there will be situations where $T$
corresponds a union of cosets of $\lat$.  In all these cases,
$\rho_{s}(T)$ converges.

The following gives the standard concentration bounds for the continuous and discrete Gaussians. \dnote{Should we prove in appendix, find exact references?}

\begin{lemma}[\cite{banaszczyk93:_new,DBLP:journals/dcg/Banaszczyk95}]
  \label{lem:gaussian-tail}
  Let $X \in \R^n$ be distributed as $D_{s}$ or $D_{\lat,s}$ for an
  $n$-dimensional lattice $\lat$. For any $\vecv \in \R^n \setminus
  \set{\veczero}$ and $t > 0$, we have
  \[
  \Pr[\pr{X}{\vecv} \geq t\length{\vecv}] \leq e^{-\pi (t/s)^2},
  \]
  and for $\eps > 0$ we have
  \[
  \Pr[\length{X}^2 \geq (1+\eps)s^2 \frac{n}{2\pi}] \leq
  ((1+\eps)e^{-\eps})^{n/2},
  \]
  which for $0 < \eps < \frac{1}{2}$ is bounded by $e^{-n \eps^{2}/6}$.
\end{lemma}

\paragraph{The smoothing parameter.}

We recall the definition of the smoothing parameter
from~\cite{DBLP:journals/siamcomp/MicciancioR07}, and define the
associated computational problem $\gapspp$.

\begin{definition}[Smoothing Parameter]
  \label{def:smoothing-param}
  For a lattice $\lat$ and real $\eps > 0$, the smoothing parameter
  $\eta_{\eps}(\lat)$ is the smallest $s > 0$ such that
  $\rho_{1/s}(\lat^* \setminus \set{\veczero}) \leq \eps$.
\end{definition}

\begin{definition}[Smoothing Parameter Problem]
  \label{def:gapsp}
  For $\gamma=\gamma(n) \geq 1$ and positive $\eps_{Y}=\eps_{Y}(n),
  \eps_{N}=\eps_{N}(n)$ with $\eps_{Y} \leq \eps_{N}$, an instance of
  $\gamma$-$\gapsp_{\eps_{Y},\eps_{N}}$ is a basis $B$ of an
  $n$-dimensional lattice $\lat=\lat(B)$.  It is a YES instance if
  $\eta_{\eps_{Y}}(\lat) \leq 1$, and is a NO instance if
  $\eta_{\eps_{N}}(\lat) > \gamma$.  When $\eps_{Y}=\eps_{N}=\eps$, we
  write $\gamma$-$\gapsp_{\eps}$.
\end{definition}

Notice that YES and NO instances are disjoint, since for a YES
instance we have $\rho(\lat^{*} \setminus \set{\veczero}) \leq
\epsilon_{Y}$, whereas for a NO instance we have $\rho(\lat^{*}
\setminus \set{\veczero}) \geq \rho_{1/\gamma}(\lat^{*} \setminus
\set{\veczero}) > \epsilon_{N} \geq \epsilon_{Y}$.

For the design and analysis of our interactive protocols, it is often
convenient to use separate $\eps_{Y}, \eps_{N}$ parameters.  The
following lemma and its corollary then let us draw conclusions about
$\gapsp$ for a single $\epsilon$ parameter, for an (often slightly)
larger approximation factor.

\begin{lemma}
  \label{lem:rho_small_eps}
  Let $\lat \subseteq \R^n$ be an $n$ dimensional lattice.  If
  $\rho_{s}(\lat \setminus \set{\veczero}) \leq \eps < 1$, then
  letting $t = \sqrt{1+\log(r)/\log(\eps^{-1})}$ for any $r \geq 1$,
  we have \[ \rho_{s/t}(\lat \setminus \set{\veczero}) \leq
  \tfrac{1}{r} \rho_{s}(\lat \setminus \set{\veczero}) \leq \epsilon/r.
  \]
\end{lemma}

\begin{proof}
  By scaling $\lat$, it suffices to prove the claim for $s=1$.  Since
  $t \geq 1$, we have
  \begin{align*}
    \rho_{1/t}(\lat \setminus \set{\veczero}) &= \sum_{\vecy \in \lat
      \setminus \set{\veczero}} e^{-\pi\length{t \vecy}^2} =
    \sum_{\vecy \in \lat \setminus \set{\veczero}}
    (e^{-\pi\length{\vecy}^2})^{t^2} \\
    &\leq \biggl( \sum_{\vecy \in \lat \setminus \set{\veczero}}
    e^{-\pi\length{\vecy}^2} \biggr)^{t^2} = \rho(\lat \setminus
    \set{\veczero})^{t^2} \leq \rho(\lat \setminus \set{\veczero})
    \cdot \eps^{t^2-1}.
  \end{align*}
  To finish the proof, note that $\eps^{t^2-1} = 1/r$, as
  needed.
\end{proof}

\begin{corollary}
  \label{cor:eps-epsY-epsN}
  For any $\eps_{N} < 1$, there is a trivial reduction from
  $\gamma'$-$\gapsp_{\eps_{Y}}$ to
  $\gamma$-$\gapsp_{\eps_{Y},\eps_{N}}$, where $\gamma'=\gamma \cdot
  \sqrt{\log(\eps_{Y}^{-1})/\log(\eps_{N}^{-1})}$.
\end{corollary}

The proof follows by a routine calculation, letting
$\epsilon=\epsilon_{N}$ and $r=\epsilon_{N}/\epsilon_{Y}$ in the above
lemma.  As a few notable examples, if $\epsilon_{Y}$ and
$\epsilon_{N}$ are both fixed constants, then the loss
$\gamma'/\gamma$ in approximation factor from the reduction is a
constant strictly greater than $1$.  But if $\epsilon_{Y}$ is constant
and $\epsilon_{N}=(1+o(1)) \cdot \epsilon_{Y}$, or if
$\epsilon_{Y}=o(1)$ and $\epsilon_{N} \leq C \cdot \epsilon_{Y}$ for a
constant $C\geq 1$, then the loss in approximation factor is only
$1+o(1)$.

\section{AM Protocol for \gapsp }
\label{sec:AM}

Here we give an Arthur-Merlin protocol for
$2$-$\gapsp_{\epsilon_{Y},\epsilon_{N}}$, defined formally in
Protocol~\ref{am-protocol}.  It is simply a Gaussian variant of the
classic Goldreich-Goldwasser
protocol~\cite{DBLP:journals/jcss/GoldreichG00}, which was originally
developed to prove that approximating the Closest and Shortest Vector
Problems to within a $c\sqrt{n/\log n}$ factor is in $\coAM$.  In our
protocol, instead of choosing an error vector $\vecx$ from the ball of
radius $c\sqrt{n/\log n}$, the verifier chooses $\vecx$ from a
continuous Gaussian distribution of parameter $1$.  It then reduces
$\vecx$ modulo the lattice (actually, the dual lattice $\lat^{*}$ in
our setting) and challenges the prover to find the original vector
$\vecx$.

For intuition on why this protocol is complete and sound, first
observe that the optimal prover strategy is maximum likelihood
decoding of the verifier's challenge $\bar{\vecx} = \vecx \bmod
\lat^{*}$, i.e., to return a most-likely element in the coset $\vecx'
\in \lat^{*}+\bar{\vecx}$.  Because the Gaussian function is
decreasing in $\length{\vecx'}$, the prover should therefore return
the shortest element $\vecx' \in \lat^{*}+\bar{\vecx}$, i.e., the
unique $\vecx' \in \calV(\lat^{*}) \cap (\lat^{*} + \bar{\vecx})$.
(We can ignore the measure-zero event that $\bar{\vecx}$ is
equidistant from two or more points in $\lat^{*}$).  The verifier can
therefore be made to accept with probability
$\gamma(\calV(\lat^{*}))$, and no more.  Note that unlike the original
Goldreich-Goldwasser protocol, ours does not have perfect
completeness, and in fact this is essential for establishing such a
small approximation factor for $\gapsp$.

For completeness, consider a YES instance where
$\eta_{\epsilon_{Y}}(\lat) \leq 1/2$, i.e., $\rho_{2}(\lat^{*}
\setminus \set{\veczero}) \leq \epsilon_{Y}$.  (For convenience, here we
scale the $2$-$\gapsp$ problem so that NO instances have
$\eta_{\epsilon_{N}}(\lat) > 1$.)  Intuitively, because the measure on
$\lat^{*} \setminus \set{\veczero}$ is small, these lattice points are
all far from the origin and so $\calV(\lat^{*})$ captures most of the
Gaussian measure $\gamma$; Lemma~\ref{lem:voronoi-gaussian} makes this
formal.  Finally, for soundness we consider the case where the
discrete measure on nonzero lattice points is relatively large, i.e.,
$\rho_{1}(\lat^{*} \setminus \set{\veczero}) > \epsilon_{N}$.
Conversely to the above, this intuitively means that the continuous
Gaussian measure $\gamma(\calV(\lat^{*}))$ cannot be too large, and
Lemma~\ref{lem:voronoi-gaussian} again makes this precise.

\begin{algorithm}
  \caption{Gaussian Goldreich-Goldwasser (GGG) Protocol}
  \label{am-protocol}
  \begin{algorithmic}[1]
    \REQUIRE Basis $B \subset \R^{n}$ of a lattice $\lat=\lat(B)$.

    \STATE Verifier chooses Gaussian $\vecx \gets D_{1}$ and sends
    $\bar{\vecx} = \vecx \bmod \lat^{*}$ to prover.

    \STATE Prover returns an $\vecx' \in \R^n$.

    \STATE Verifier accepts if $\vecx' = \vecx$.
  \end{algorithmic}
\end{algorithm}

\begin{theorem}
  \label{thm:GGG}
  For $0 < \eps \leq \delta < \frac{1}{2}$, Protocol~\ref{am-protocol}
  on lattice $\lat=\lat(B)$ satisfies:
  \begin{enumerate}
  \item \emph{Completeness:} If $\eta_{\eps}(\lat) \leq \frac{1}{2}$,
    then there exists a prover that makes the verifier accept with
    probability at least $1-\eps$.
  \item \emph{Soundness:} If $\eta_{\frac{\delta}{1-\delta}}(\lat)
    \geq 1$, then the verifier rejects with probability at least
    $\delta$.
  \end{enumerate}
  In particular, $2$-$\gapsp_{\epsilon, \delta/(1-\delta)} \in \AM$
  when $\delta-\epsilon \geq 1/\poly(n)$.  Moreover, when
  $\epsilon=\negl(n)$ the protocol is honest-verifier statistical
  zero knowledge, i.e., $2$-$\gapsp_{\epsilon,\delta/(1-\delta)} \in
  \HVSZK = \SZK$.
\end{theorem}

By applying Corollary~\ref{cor:eps-epsY-epsN}, we obtain the following
upper bounds on the complexity of $\gamma$-$\gapspp_{\eps}$ for
different ranges of $\epsilon$. \knote{``$1/\poly(n) \leq \eps(n) \leq
  o(1)$'' in the following corollary statement is less formal, but
  makes the statement cleaner. Better ways?}

\begin{corollary}
  For the following $\eps(n) < 1$, we have
  $\gamma$-$\gapspp_{\epsilon} \in \AM$ for the following $\gamma(n)$:
  \begin{itemize}[itemsep=0pt]
  \item If $\eps(n) \leq \negl(n)$, then
    $\gamma=O(\sqrt{\log(\eps^{-1})/\log n})$.
  \item If $1/\poly(n) \leq \eps(n) \leq o(1)$, then $\gamma=(2+o(1))$.
  \item If $\eps(n) \geq \Omega(1)$, then $\gamma=O(1)$.
  \end{itemize}
\end{corollary}

\noindent
The next two lemmas provide the crux of the proof of
Theorem~\ref{thm:GGG}.

\begin{lemma}
  \label{lem:gaussian-translate}
  Let $S \subseteq \R^n$ be symmetric (i.e., $S=-S$) measurable
  set. Then for any $\vecy \in \R^n$, \[ \gamma_s(S+\vecy) \geq
  \gamma_s(S) \cdot \rho_s(\vecy). \]
\end{lemma}

\begin{proof}
  By scaling $S$ and $\vecy$, it suffices to prove the claim for
  $s=1$.  For any $t \in \R$, note that $\cosh(t) = \frac12
  (e^{t}+e^{-t}) \geq 1$.  We have
  \begin{align*}
    \gamma(S+\vecy) &= \int_{S} e^{-\pi\length{\vecy-\vecx}^2}\, dx =
    \int_S \frac{1}{2}(e^{-\pi\length{\vecy-\vecx}^2} +
    e^{-\pi\length{\vecy+\vecx}^2})\, dx
    & \text{(symmetry of $S$)}\\
    &= e^{-\pi\length{\vecy}^2} \int_S e^{-\pi\length{\vecx}^2} \cdot
    \frac12 \left(e^{2\pi\pr{\vecx}{\vecy}} +
      e^{-2\pi\pr{\vecx}{\vecy}}\right)\, dx
    & \text{(expanding the squares)} \\
    &\geq \rho(\vecy) \int_S \rho(\vecx)\, dx = \rho(\vecy) \cdot
    \gamma(S).
  \end{align*}
\end{proof}

The following crucial lemma establishes a tight relationship between
discrete Gaussian sums on $\lat$ and the Gaussian mass of the Voronoi
cell. 

\begin{lemma}[Voronoi Cell Characterization]
  \label{lem:voronoi-gaussian}
  Let $\lat \subseteq \R^n$ be a lattice with Voronoi cell
  $\calV=\calV(\lat)$, and let $s > 0$.  Then
  \[ \frac{\rho_s(\lat \setminus \set{\veczero})}{\rho_s(\lat)} \leq
  1-\gamma_s(\calV) \leq \rho_{2s}(\lat \setminus \set{\veczero}). \]
  In particular, letting $s_\eps = \eta_\eps(\lat^*)$ for some $\eps
  \in (0,1)$, we have that $\gamma(2 s_\eps \calV) \geq 1-\eps$ and
  $\gamma(s_\eps \calV) \leq \frac{1}{1+\eps}$.
\end{lemma}

\begin{proof}
  By scaling $\lat$, it suffices to prove the claim for $s=1$.  We
  first show the upper bound.  Let $X \in \R^n$ be distributed as
  $D_{1}$, and note that $1-\gamma(\calV) = \Pr[X \notin \calV]$.  By
  the union bound and Lemma~\ref{lem:gaussian-tail},
  \begin{align*}
    \Pr[X \notin \calV] &= \Pr[\bigcup_{\vecy \in \lat \setminus
      \set{\veczero}} \set{\pr{X}{\vecy} > \tfrac12
      \pr{\vecy}{\vecy}}] \leq \sum_{\vecy \in \lat \setminus
      \set{\veczero}}
    \Pr[\pr{X}{\vecy} > \tfrac{1}{2} \pr{\vecy}{\vecy}] \\
    &\leq \sum_{\vecy \in \lat \setminus \set{\veczero}} e^{-\pi
      \length{\vecy/2}^2} = \rho_{2}(\lat \setminus \set{\veczero}).
  \end{align*}
  We now prove the lower bound.  Since $\calV$ tiles space with
  respect to $\lat$, by applying Lemma \ref{lem:gaussian-translate} with
  $S = \calV$, we have
  \[
  1-\gamma(\calV) = \gamma(\R^n \setminus \calV) = \sum_{\vecy \in
    \lat \setminus \set{\veczero}} \gamma(\calV + \vecy) \geq
  \gamma(\calV) \cdot \rho(\lat \setminus \set{\veczero}),
  \]
  Rearranging terms and using $\rho(\set{\veczero})=1$, we have
  $1-\gamma(\calV) \geq 1-1/\rho(\lat) = \rho(\lat \setminus
  \set{\veczero}) / \rho(\lat)$, as desired. 
  Finally, the ``in particular'' claim follows from $\gamma(r \calV) =
  \gamma_{1/r}(\calV)$ and an easy calculation.
%
\end{proof}

\begin{proof}[Proof of Theorem~\ref{thm:GGG}]
  As already argued above, the optimal prover strategy given
  $\bar{\vecx} \in \R^{n}$ is maximum likelihood decoding, and the
  optimal prover can make the verifier accept with probability
  $\gamma(\calV(\lat^{*}))$.  Completeness and soundness now follow
  immediately from Lemma~\ref{lem:voronoi-gaussian}, as already
  outlined in the overview.

  For honest-verifier statistical zero-knowledge when
  $\epsilon=\negl(n)$, the simulator just chooses $\vecx \gets D_{1}$
  as the verifier's randomness, and outputs $\vecx$ as the message
  from the prover.  Because the prover also returns~$\vecx$ with
  probability at least $1-\epsilon$ in the real protocol, the
  simulated transcript is within negligible statistical distance of
  the real transcript.
\end{proof}

\section{SZK Protocol for GapSPP}
\label{sec:SZK}

This section is devoted to showing that $(2+ o(1))$-$\gapsp_{1/\poly(n)}$ is in $\SZK$. %

\begin{theorem} \label{thm:SZK-main} For every $\eps:\N \rightarrow [0,1]$ such that $\frac{1}{\poly(n)} \leq \eps(n) \leq \frac{1}{36}$, we have $2\cdot(1+\delta)$-$\gapsp_{\eps,12\eps} \in \SZK$, where $\delta = \sqrt{\frac{3}{2n} \ln \frac{4}{\eps}}$.
\label{thm:szk}
\end{theorem}

As before, the following corollary gives the implied upper bound on the complexity of $\gamma$-$\gapspp_{\eps}$ (by applying Corollary~\ref{cor:eps-epsY-epsN}).

\cnote{State and format this one like the corollary in previous section?}
\begin{corollary}
  For every $\eps: \N \rightarrow (0,1)$, if $1/\poly(n) \leq \eps(n) \leq o(1)$, then $(2+o(1))$-$\gapspp_{\eps} \in \SZK$. If $\eps(n)\leq \negl(n)$, then $O\left(\sqrt{\frac{\log(1/\eps)}{\log n}} \right)$-$\gapsp_{\eps}  \in \SZK$. Finally, if $\Omega(1) \leq \eps(n) \leq 1/3$, then $O(1)$-$\gapsp_{\eps} \in \SZK$.
\end{corollary}

Our construction follows a classic approach of constructing an
\emph{instance-dependent ($\ID$) commitment scheme} for $\gapsp$,
which is known to be sufficient for obtaining a $\SZK$
protocol~\cite{ItohOS97}. With an additional observation, we show that
a \emph{significantly weaker} notion of $\ID$ commitment schemes is
sufficient to obtain $\SZK$ protocols; roughly speaking, we only need
an $\ID$ bit-commitment scheme that is sufficiently binding for an
\emph{honest} sender, and hiding (from a dishonest
receiver). Specifically, it is sufficient to construct a
``non-trivial'' $\ID$ commitment scheme defined as follows.

\begin{definition} \label{def:ID-Com} Let $\Pi$ be a promise problem. A (non-interactive) instance-dependent bit-commitment scheme $\Com$ for $\Pi$ is a $\PPT$ algorithm that on input an instance $x \in \zo^n$ and a bit $b\in \zo$, outputs a commitment $\Com_x(b) \in \zo^*$. Let $p = p(n), q = q(n) \in (0,1)$. We define (weak) binding and hiding properties of $\Com$ as follows.
\begin{itemize}
\item \emph{Statistical honest-sender $q$-binding for YES instances:} For every $x\in \Pi_Y$ and $b \in \zo$, 
\[ \Pr[ \Com_x(b) \in \supp(\Com_x(\bar{b})) ] \leq q(|x|). \]
(Note that when $\Com_x(b) \notin \supp(\Com_x(\bar{b}))$, the
commitment $\Com_x(b)$ cannot be opened to $\bar{b}$. Thus, the above
condition implies that the binding property can be broken with probability at most $q$.)
\item \emph{Statistical $p$-hiding for NO instances:} For every $x \in \Pi_N$, 
\[ \Delta( \Com_x(0), \Com_x(1)) \leq p(|x|). \]
(The above condition implies that given $\Com_x(b)$ for a random $b$,
one can only predict $b$ correctly with probability at most $(1+p)/2$,
which means that the hiding property can be broken with advantage at most $p$.)
\end{itemize}
$\Com$ is \emph{non-trivial} if $\Com$ is statistical $p$-hiding and statistical honest-sender $q$-binding with $p + q \leq 1 - 1/\poly(n)$.\footnote{This is in contrast to the fact that one can construct a (trivial) $p$-hiding and $q$-binding commitment scheme for every $p+q \geq 1$. For example, defining $\Com_x(b) = b$ gives $p=1$ and $q=0$, and defining $\Com_x(b) = 0$ gives $p=0$ and $q=1$. More generally, defining $\Com_x(b)$ to be $b$ with probability $\alpha$ and $00$ with probability $1-\alpha$ gives $p=\alpha$ and $q=1-\alpha$.} $\Com$ is $\emph{secure}$ if $\Com$ is statistical $p$-hiding and statistical honest-sender $q$-binding with negligible $p$ and $q$. 
\end{definition}

In the next subsection, we focus on constructing a non-trivial $\ID$ commitment schemes for $(2+o(1))$-$\gapsp_{1/\poly(n)}$. We present more detailed background for $\ID$ commitment schemes and discuss why it is sufficient to construct $\SZK$ protocols in Section~\ref{subsec:overview-SZK}.

\subsection{A Non-Trivial $\ID$ Commitment Scheme for GapSPP}
\label{subsec:ID-Com-for-GapSP}

In this section, we construct a non-trivial $\ID$ bit-commitment
scheme $\SPCom$ for $(2+ o(1))$-$\gapsp_{1/\poly(n)}$. Our
construction can be viewed as a generalization of an
instance-dependent commitment scheme implicit
in~\cite{DBLP:conf/crypto/MicciancioV03} for $O(\sqrt{n/\log
  n})$-$\gapsvp$.\footnote{While~\cite{DBLP:conf/crypto/MicciancioV03}
  constructed their protocol by combining the reduction from $\gapsvp$
  to $\gapcvp$ with Goldreich-Levin hardcore predicate, their
  construction can be interpreted as implicitly constructing an $\ID$
  bit-commitment scheme for $\gapsvp$ by first constructing one with
  \emph{perfect} binding but weak hiding, and then amplifying the
  hiding property.}
To explain the intuition behind our construction, it is instructive to first consider the construction of $\ID$ commitment scheme for $\gapsvp$ (for simplicity, below we describe commitment to a random $b$): 
To commit, a sender first selects a ``random'' lattice point $\vecv \in \lat$ (see Figure~\ref{alg:commit-scheme} for the precise distribution)
and adds a random noise vector $\vece$ drawn from a ball of certain radius (say, $r = 1/2$) to $\vecv$; let $\vecw = \vecv + \vece$. Intuitively, the
vector $\vecw$ is binding to $\vecv$ if the noise is sufficiently short. To actually commit to a bit, the sender also samples a random hash function
$h$, and commits to the hashed bit $b = h(\vecv)$. Namely, $(\vecw, h)$ is a commitment $\Com_{\lat}(b)$ to $b= h(\vecv)$.

Intuitively, if the length of the shortest vector $\lambda_1(\lat) \geq 1$, then all balls centered at lattice points $\vecv \in \lat$ of radius
$r=1/2$ are disjoint, and thus $\Com_{\lat}(b) = (\vecw, h)$ is perfect binding. On the other hand, if the shortest vector is too short, say,
$\lambda_1(\lat) \leq O(\sqrt{(\log n)/n})$, then $\vecw$ may fall in the intersection region of two (or more) balls with non-negligible probability,
using the symmetry of the lattice and the fact that the balls centered around the origin and a shortest non-zero vector have non-negligible overlap. When $\vecw$ lies in the
balls centered at $\vecv_1$ and $\vecv_2$ and $h(\vecv_1) \neq h(\vecv_2)$, the commitment $\Com_{\lat}(b) = (\vecw,h)$ does not reveal the committed
value $b$, which intuitively achieves hiding. Indeed, the above argument can be formalized readily, yielding an $\ID$ bit-commitment scheme for
$O(\sqrt{n/\log n})$-$\gapsvp$ with perfect biding and weak hiding properties.

Note that in the above commitment scheme, the quality of the hiding property depends on how much the ball $\vecv+rB_2^n$ overlaps with the balls
around surrounding lattice points. However, in the above analysis, we
only exploit the overlap contributed by a nearest lattice point to
$\vecv$,
ignoring the overlap contributed by all other balls. In general, such an approach can only give
a very coarse approximation of the overlap, which one can see from the example of extremal lattices where there are exponentially many lattice points
of length roughly equal to that of the shortest vector. As a result, using this approach one can only obtain a non-trivial $\ID$ bit-commitment scheme
for $\gamma$-$\gapsvp$ with $\gamma \geq \Omega(\sqrt{n/\log n})$.

Our key observation is that, when we switch from $\gapsvp$ to $\gapspp$, the above construction gives a non-trivial $\ID$ bit-commitment scheme for
$\gamma$-$\gapspp_{1/\poly(n)}$ with $\gamma = 2 + o(1)$. This stems from our new ball overlap characterization of the smoothing parameter, which
gives us much finer control on the amount overlap we obtain in the above protocol. We formalize this characterization as follows:

\newcommand{\overlap}{\mathrm{Overlap}}

\begin{lemma}[Ball Overlap Characterization]
  \label{lem:ball-overlap-char} Let $\lat$ be an $n$ dimensional
  lattice. For $r > 0$, define
  \[
  \overlap(\lat,r) \eqdef \frac{\vol_n\left(\bigcup_{\vecy \in \lat
        \setminus \set{\veczero}} \left(rB_2^n \cap \left(rB_2^n +
          \vecy\right)\right)\right)}{\vol_n(rB_2^n)},
  \]
  which denotes the fraction of overlap of a ball of radius $r$
  centered at a point in $\lat$ with balls of equal radius centered at
  all other lattice points.  Then for $\eps \in (2^{o(-n)},1/3)$,
  setting $r_\eps =\sqrt{\frac{n}{2\pi}} / (2\eta_{\eps}(\lat^*))$,
  the following holds:\knote{Changed definition of $r_\eps$ by a
    factor of $2$
    to make it consistent with the intro and a bit more clean.}
  \begin{enumerate}[itemsep=0pt]
  \item For $0 \leq r \leq r_{\eps}$, we have $\overlap(\lat,r) \leq
    2\eps$.
  \item For any $r \geq 2(1+\delta) \cdot r_\eps$ where $\delta =
    \sqrt{\frac{3}{2n} \ln \frac{4}{\eps}}$, we have $\overlap(\lat,
    r) \geq \eps/2$.
  \end{enumerate}
\end{lemma}

The above lemma says that up to a factor of $2+o(1)$, the smoothing
parameter $\eta_{\eps}(\lat^*)$ characterizes the required radius for
balls on $\lat$ to have roughly $\eps$ fraction of overlap. As we
shall see shortly, the amount of overlap tightly characterizes the
binding and hiding property of the commitment scheme described
above. As such, by choosing $\eps_Y$ and $\eps_N$ with a small
constant factor gap, the above construction yields a non-trivial $\ID$
bit-commitment scheme for $\gamma$-$\gapspp_{\eps_Y,\eps_N}$ with
$\gamma = 2+o(1)$.

We proceed to formally define our $\ID$ bit-commitment scheme $\SPCom$ for $\gapspp$ in Fig~\ref{fig:SPCom}, and establish its binding and hiding
properties. We prove the binding and hiding properties in Lemma~\ref{lem:binding} and~\ref{lem:hiding}, respectively, and summarize the properties
of $\SPCom$ in Lemma~\ref{lem:SPCom-non-trivial}. We defer the proofs of all geometric lemmas (in particular, the Ball Overlap Characterization) to
subsection~\ref{sec:szk-geom}.

We remark that since we are approximating $\eta_\eps(\lat)$, the following protocol operates directly on $\lat^*$. For simplicity of notation, for a
basis $B$ of $\lat$, we write $B^* = (B^{-1})^t$
\knote{Just to double check, is $B^{-t}$ typo or standard notation?}
to denote the corresponding dual basis of $\lat^*$.
\cnote{$B^{-t}$ is semi-standard but since it is so rare in this paper
we can just write $(B^{-1})^{t}$ to avoid confusion.}

\begin{figure}[h]
\label{alg:commit-scheme}
\begin{center}
\begin{tabular}{|p{16cm}|c|}
\hline 
Let $\mathcal{H} = \{ h : \zo^n \rightarrow \zo\}$ be a pairwise-independent hash family. \\
On input a lattice basis $B$ and a bit $b \in \zo$,
\begin{itemize}[itemsep=0pt]
\item Sample uniformly random $\vecz \leftarrow \zo^n$ and $h \leftarrow \mathcal{H}$ jointly subject to $h(\vecz)=b$. (This can be done by rejection sampling, or equivalently by sampling uniform $\vecz \leftarrow \zo^n$ first, and then sampling $h \leftarrow \mathcal{H}$ conditioned on $h(\vecz)=b$.)
\item Sample $\vece \leftarrow rB^n_2$ with $r = \frac{1}{2}\sqrt{\frac{n}{2\pi}}$.
\item Let $\vecv = B^*\vecz$ and $\vecw = (\vecv+\vece \bmod 2B^*)$.
\item Output $\SPCom_{B}(b) = (\vecw, h)$.
\end{itemize}
\\ 
\hline
\end{tabular}
\end{center}
\caption{\label{fig:SPCom} $\SPCom$: a non-trivial $\ID$ commitment scheme for $\gapsp$.
}
\end{figure}

The following two technical lemmas establish the (weak) binding and hiding properties of $\SPCom$.

\begin{lemma} \label{lem:binding} For every $b \in \zo$, %
\[ \Pr[ \SPCom_{B}(b) \in \supp( \SPCom_{B}(\bar{b})) ] \leq \overlap(\lat^*, r).\]
\end{lemma}
\begin{proof} Let $S = \bigcup_{\vecy \in \lat^* \setminus \set{\veczero}} \left(rB_2^n \cap (rB_2^n + \vecy)\right)$. By definition, $\SPCom_{B}(b)$ is generated by sampling $\vece\leftarrow rB^n_2$ and $h\leftarrow \mathcal{H}, \vecz \in \zo^n$ such that $h(\vecz)=b$, and outputting $(\vecw,h) = (\vecv + \vece \bmod 2B^*, h)$, where $\vecv = B^*\vecz$. Thus, we can write 
$\vecw = \vecv + \vecu + \vece$ for some $\vecu \in 2\lat^*$.

The event $\SPCom_{B}(b) \in \supp( \SPCom_{B}(\bar{b}))$ means that there exists some $\vecz' \in \zo^n, \vece' \in rB^n_2$ such that $h(\vecz') = \bar{b}$ and $w = (\vecv' + \vece' \bmod 2B^*)$, where $\vecv' = B^*\vecz'$. Similarly, we can write $\vecw = \vecv' + \vecu' + \vece'$ for some $\vecu' \in 2\lat^*$.

Let $\vecy = \vecv' + \vecu' - \vecv - \vecu$, and note that $\vecy \in \lat^*$. The facts that $\vecw = \vecv+\vecu+\vece = \vecv'+\vecu'+\vece'$ and $\vece' \in rB^2_n$ imply that $\vece \in rB^n_2 + \vecy$, which implies $\vece \in S$. As the event in the LHS implies  $\vece \in S$, it follows that
\[ \Pr[ \SPCom_{B}(b) \in \supp( \SPCom_{B}(\bar{b})) ] \leq \Pr[\vece
\in S] = \frac{\vol_n(S)}{\vol_n(rB_2^n)} = \overlap(\lat^*, r). \]
\end{proof}

\begin{lemma} \label{lem:hiding}
\[ \Delta( \SPCom_{B}(0), \SPCom_{B}(1) ) \leq 1 -
(\overlap(\lat^*,r)/2). \]
\end{lemma}
\begin{proof} Let $S = \bigcup_{\vecy \in \lat^* \setminus \set{\veczero}} \left(rB_2^n \cap (rB_2^n + \vecy)\right)$. Define random variables $(W_0,H_0) = \SPCom_{B}(0)$ and $(W_1,H_1) =  \SPCom_{B}(1)$. Observe that the marginal distributions of $W_0$ and $W_1$ are identical (following by the fact that $h\leftarrow \mathcal{H}$ maps every $\vecz\in \zo^n$ to a uniformly random bit), we have
\begin{eqnarray*}
\Delta( \SPCom_{B}(0), \SPCom_{B}(1) )
 & = & (1/2) \cdot \sum_{\vecw,h} \left| \Pr_{W_0,H_0}[(\vecw,h)] - \Pr_{W_1,H_1}[(\vecw,h)] \right| \\
 & = & (1/2) \cdot \sum_{\vecw} \Pr_{W_0}[\vecw] \cdot \sum_{h} \left| \Pr_{H_0}[h | W_0=\vecw] -  \Pr_{H_1}[h | W_1=\vecw] \right| \\
 & = & \E_{\vecw\leftarrow W_0} [ \Delta( H_0|_{W_0=\vecw}, H_1|_{W_1=\vecw}) ]
\end{eqnarray*}

For every $\vecw \in (\R^n \bmod 2B^*)$, define $T_\vecw = (rB^n_2 + \vecw) \cap \lat^*$ and $T'_\vecw = (T_\vecw \bmod 2B^*)$. we rely on the following
two technical claims to upper bound the statistical distance. Note that the event $|T_\vecw| \geq 2$ is equivalent to the event $\vece \in S$, where
$\vece$ is the error vector used to generate $\vecw$, and hence 
\[
\Pr[|T_\vecw| \geq 2] = \Pr[\vece \in S] = \overlap(\lat^*,r)\text{.}
\]

\begin{claim} \label{clm:mod_is_OK}
\[ \Pr_{\vecw\leftarrow W_0}[ |T'_\vecw| \geq 2] =
\Pr_{\vecw\leftarrow W_0}[ |T_\vecw| \geq 2]. \]
\end{claim}

\begin{claim} \label{clm:stat_dist} For every $\vecw \in (\R^n \bmod 2B^*)$ with $|T'_\vecw| \geq 2$, 
\[ \Delta( H_0|_{W_0=\vecw}, H_1|_{W_1=\vecw}) \leq 1/2. \]
\end{claim}

The above two claims imply that
\begin{eqnarray*}
\Delta( \SPCom_{B}(0), \SPCom_{B}(1) ) 
 & \leq & \Pr_{\vecw\leftarrow W_0}[ |T'_\vecw| \geq 2] \cdot (1/2) + \Pr_{\vecw\leftarrow W_0}[ |T'_\vecw| = 1] \cdot 1\\
 & = & \overlap(\lat^*,r) \cdot (1/2) + ( 1- \overlap(\lat^*,r)) \\
 & = & 1 - \overlap(\lat^*,r)/2,
\end{eqnarray*}
as desired. It remains to prove the claims.

\begin{proof} (of Claim~\ref{clm:mod_is_OK}) Let $\set{\vecy_1,\dots,\vecy_t} = T_\vecw = \lat^* \cap (rB2^n + \vecw)$, where the $\vecy_i$s are
ordered such that $\|\vecy_i-\vecw\|_2 \leq \|\vecy_{i+1}-\vecw\|_2$. By assumption, we have that $t \geq 2$. To prove that $|T_\vecw'| = |T_\vecw \pmod{2B^*}| \geq
2$, it suffices to show that $\vecy_1 \neq \vecy_2 \pmod{2B^*}$. Assume not, then note that $\bar{\vecy} = \frac{1}{2}(\vecy_1 + \vecy_2) \in \lat^*$, $\bar{\vecy}
\neq \vecy_1$, $\bar{\vecy} \neq \vecy_2$. Furthermore, by the triangle inequality
\[
\|\bar{\vecy}-\vecw\|_2 = \|\frac{1}{2}(\vecy_1 + \vecy_2)-\vecw\|_2 \leq \frac{1}{2}\|\vecy_1-\vecw\|_2 + \frac{1}{2}\|\vecy_2-\vecw\|_2 \leq \|\vecy_2-\vecx\|_2 \leq r
\] 
Hence $\bar{\vecy} \in T_\vecw$. We now examine two cases. If $\|\vecy_1-\vecw\|_2 = \|\vecy_2-\vecw\|_2$, then since $\vecy_1 \neq \vecy_2$, the above inequality must hold strictly. But then
$\|\bar{\vecy}-\vecw\|_2 < \|\vecy_1-\vecw\|_2$, which contradicts the fact that $\vecy_1$ is a closest lattice vector to $\vecw$. If $\|\vecy_1-\vecw\|_2 < \|\vecy_2-\vecw\|_2$, then
$\|\bar{\vecy}-\vecw\|_2 < \|\vecy_2-\vecw\|_2$, which contradicts that $\vecy_2$ is a closest lattice vector to $\vecw$ after $\vecy_1$. The claim thus follows.
\end{proof}

\begin{proof} (of Claim~\ref{clm:stat_dist}) Let $T'_\vecw = \{\vecv_1,\dots,\vecv_t\}$ and let $\vecz_i$ be the coordinates of $\vecv_i$ with respect to the basis $B^*$ for every $i \in [t]$. Note that by construction, conditioned on $\vecw$, the random variable $\vecz \in \zo^n$ becomes uniform over the $\set{\vecz_1,\dots,\vecz_t}$. 

Now, consider a probability space $\mathcal{P}$ defined by independent random variables $(I,H)$, where $I$ is a uniformly random index in $[t]$ and $H$ is a random hash function in $\mathcal{H}$. Define a random variable $B = H(\vecz_{I})$. Note that by the construction, for $b \in \zo$, the random variable $H_b|_{W_b=\vecw}$ has identical distribution to the random variable $H|_{B=b}$ in $\mathcal{P}$. Thus, our goal can be rephrased as to upper bound $\Delta( H|_{B=0}, H|_{B=1})$.

By Bayes' rule, 
\[ \Pr[H|_{B=b} = h] = \frac{\Pr[H=h] \Pr[B=b| H=h]}{Pr[B=b]} = 2
\Pr[H=h] \cdot \frac{\#\{i: h(\vecz_i)=b\}}{t}, \]
thus
\[ \Delta( H|_{B=0}, H|_{B=1}) = \sum_h \Pr[H=h] \frac{| \#\{i:
  h(\vecz_i)=0\} - \#\{i: h(\vecz_i)=1\}|}{t}. \]

Intuitively, since $\mathcal{H}$ is a pairwise-independent hash-family, the discrepancy $\#\{i: h(\vecz_i)=0\} - \#\{i: h(\vecz_i)=1\}$ should be
small on expectation. We prove this presently.

\begin{claim} For $t \geq 2$, $\E_{h \leftarrow \mathcal{H}}[|\#\{i: h(\vecz_i)=0\} - \#\{i: h(\vecz_i)=1\}|] \leq \frac{t}{2}$ \end{claim}
\begin{proof}
For $i \in [t]$, let $X_i = (-1)^{h(\vecz_i)} \in \set{-1,1}$. Since $\Pr[h(\vecz_i)=0]=\Pr[h(\vecz_i)=1]=1/2$, we have that $\E[X_i] = 0$ for all $i \in [t]$. 
Furthermore, by pairwise independence we have that $\E[X_iX_j] = \E[X_i]\E[X_j] = 0$ for distinct $i,j \in [t]$. By definition of $X_i$, it is
easy to verify that 
\[ |\#\{i: h(\vecz_i)=0\} - \#\{i: h(\vecz_i)=1\}| = |\sum_{i=1}^t
X_i|. \]

\noindent By Jensen's inequality and pairwise independence, we obtain the inequality
\[
\E\left[\left|\sum_{i=1}^t X_i\right|\right]^2 \leq \E\left[\left(\sum_{i=1}^t X_i\right)^2\right] = \sum_{1 \leq i,j \leq t} \E[X_iX_j] = \sum_{i=1}^t \E[X_i^2] = t
\]
Taking a square root, the above inequality gives us $\E[|\#\{i: h(\vecz_i)=0\} - \#\{i: h(\vecz_i)=1\}|] \leq \sqrt{t}$. Since
$\sqrt{t} \leq t/2$ for $t \geq 4$, the claim holds for all $t \geq 4$.

\noindent It remains to prove the claim for $t=2,3$. For $t=2$, $h$ acts like a truly random hash function, and hence a direct computation yields
\begin{equation}
\label{eq:szk-sd-1}
\E[|X_1+X_2|] = 2 \Pr[X_1=X_2] + 0 \Pr[X_1 \neq X_2] = 2 (1/2) + 0 = 1,
\end{equation}
as needed. For the case $t=3$, we have that
\begin{align*}
\E[|X_1+X_2+X_3|] &= 3 \Pr[X_1=X_2=X_3] + 1 \Pr[X_1,X_2,X_3 \text{ not all equal }] \\
                  &= 3 \Pr[X_1=X_2=X_3] + 1(1 - \Pr[X_1=X_2=X_3]) = 1 + 2\Pr[X_1=X_2=X_3] \text{.}
\end{align*}
By inclusion exclusion we get that
\begin{align*}
1 &= \Pr[\exists X_i = 1] + \Pr[X_1=X_2=X_3=-1] \\
  &= \sum_{i=1}^3 \Pr[X_i=1] - \sum_{1\leq i<j\leq 3} \Pr[X_i=X_j=1] + \Pr[X_1=X_2=X_3=1] + \Pr[X_1=X_2=X_3=-1] \\
  &= \sum_{i=1}^3 \Pr[X_i=1] - \sum_{1\leq i<j\leq 3} \Pr[X_i=X_j=1] + \Pr[X_1=X_2=X_3]
\end{align*}
By rearranging the above equality and using pairwise independence, we get
\begin{equation}
\label{eq:szk-sd-2}
\Pr[X_1=X_2=X_3] = 1 - \sum_{i=1}^3 \Pr[X_i=1] + \sum_{1 \leq i<j \leq 3} \Pr[X_i=X_j=1] = 1 - 3(1/2) + 3(1/4) = \frac{1}{4}
\end{equation}
Combining Equations \eqref{eq:szk-sd-1} and \eqref{eq:szk-sd-2}, we get that $\E[|X_1+X_2+X_3|] = 1 + 2(1/4) = 3/2$, as needed.
\end{proof}

From the above claim, we observe that $\Delta( H|_{B=0}, H|_{B=1}) \leq \frac{t/2}{t} \leq 1/2$ for every $t \geq 2$ as needed.
\end{proof}

\end{proof}

Finally, we prove the $\ID$ binding and hiding properties of $\SPCom$ by Lemma~\ref{lem:ball-overlap-char},~\ref{lem:binding}, and~\ref{lem:hiding}.

\begin{lemma} \label{lem:SPCom-non-trivial} For every $\eps:\N \rightarrow [0,1]$ such that $1/\poly(n) \leq \eps(n) \leq 1/36$, $\SPCom$ is a non-trivial $\ID$ commitment scheme for $2\cdot(1+\delta)$-$\gapsp_{\eps,12\eps}$ with $\delta = \sqrt{\frac{3}{2n} \ln \frac{4}{\eps}}$. Specifically, $\SPCom$ is $(2\eps)$-binding for the YES-instances and $(1-3\eps)$-hiding for the NO-instances of $2\cdot(1+\delta)$-$\gapsp_{\eps,12\eps}$, respectively.
\end{lemma}
\begin{proof} For YES-instances where $\eta_{\eps}(\lat) \leq 1$, by Part 1. of Lemma~\ref{lem:ball-overlap-char} and noting that $r \geq r_\eps$,
\[ \overlap(\lat,r) \leq 2\eps. \]
Thus, by Lemma~\ref{lem:binding}, $\SPCom$ is $(2\eps)$-binding for the YES-instances. On the other hand, for NO-instances where $\eta_{\eps}(\lat) \geq 2\cdot (1+\delta)$, by Part 2. of Lemma~\ref{lem:ball-overlap-char} and noting that $r \geq 2\cdot(1+\delta)\cdot r_{\eps}$,
\[ \overlap(\lat,r) \geq 12\eps/2 = 6\eps. \]
Thus, by Lemma~\ref{lem:hiding}, $\Com$ is $(1-3\eps)$-hiding for the NO-instances.
\end{proof}

Theorem~\ref{thm:SZK-main} then follows by combining Lemma~\ref{lem:SPCom-non-trivial} and Theorem~\ref{thm:Com-amplification} stated in the next section. We remark that our $\SZK$ protocol for $(2+o(1))$-$\gapsp_{1/\poly(n)}$ does not have efficient prover strategy, since we do not know if the problem is in $\NP$ or $\coNP$.

\subsection{Geometric Lemmas}
\label{sec:szk-geom}

Here we prove several geometric lemmas with the goal of establishing the ball overlap characterization (Lemma~\ref{lem:ball-overlap-char}). The first lemma gives a standard upper bound on the volume of the intersection of two euclidean balls (see \cite[Lemma 2.2]{Ball97}, noting that the
ball intersection volume is at most twice that of the spherical cap). 

\begin{lemma}
  \label{lem:sphere-int}
  For $r > 0$, let $s = r \sqrt{2\pi/n}$. Then for any $\vecy \in \R^n$,
  \[
  \frac{\vol_n(rB_2^n \cap (\vecy + rB_2^n))}{\vol_n(rB_2^n)} \leq 2
  e^{-\pi\length{\frac{\vecy}{2s}}^2}.
  \]
\end{lemma}

The following lemma will allows us to transfer lower bounds on the gaussian measure of the overlap region to lower bounds on uniform measure, and will
be important in the proof of the ball overlap characterization.

\begin{lemma}
  \label{lem:gaus-unif-trans}
  Let $K \subseteq \R^n$ be a convex body containing the origin. Then for any $r,s > 0$ we have that
  \[
  \frac{\gamma_s(rB_2^n \setminus K)}{\gamma_s(rB_2^n)} \leq \frac{\vol_n(rB_2^n \setminus K)}{\vol_n(rB_2^n)}
  \]
\end{lemma}

\begin{proof}
  The plan of the proof is to show that since gaussian measure is more
  biased towards the origin than the uniform measure, switching from
  gaussian to uniform pushes measures outside of $K$. To make this
  precise, we note that
  \begin{equation}
    \label{eq:gut1}
    \frac{\gamma_s(rB_2^n \setminus K)}{\gamma_s(rB_2^n)} = \int_{S^{n-1}} \int_0^r I[t\theta \notin K] \frac{e^{-\pi\frac{t}{s}^2}}{\gamma_s(rB_2^n)s^n} t^{n-1} dtd\theta
  \end{equation}
  where $d\theta$ is the Haar measure on the unit sphere $S^{n-1}
  \subseteq \R^n$. Furthermore, note that
  \begin{equation}
    \label{eq:gut2}
    \frac{\vol_n(rB_2^n \setminus K)}{\vol_n(rB_2^n)} = \int_{S^{n-1}} \int_0^r I[t\theta \notin K] \frac{1}{\vol_n(rB_2^n)} t^{n-1} dt d\theta
  \end{equation}
  Since both the uniform and gaussian measure are spherically
  symmetric, we must have that
  \begin{equation}
    \label{eq:gut3}
    \int_0^r \frac{1}{\vol_n(rB_2^n)} t^{n-1} dt = \int_0^r \frac{e^{-\pi(t/s)^2}}{s^n\gamma_s(rB_2^n)} t^{n-1}dt 
  \end{equation}
  Let $f,g:[0,r] \rightarrow \R_+$ be defined by $f(t) =
  \frac{1}{\vol_n(rB_2^n)}$ and $g(t) =
  \frac{e^{-\pi(t/s)^2}}{s^n\gamma_s(rB_2^n)}$.  Since $f$ is
  constant, $g$ is decreasing, and $\int_0^r f(t)t^{n-1}dt = \int_0^r
  g(t)t^{n-1}dt$ (by equation \eqref{eq:gut3}), there exists $b \in
  (0,r)$ such that $f(t) \leq g(t)$ on $[0,b]$ and $f(t) \geq g(t)$ on
  $(b,r]$. Given this geometry, we clearly have that for all $c \in
  [0,r]$
  \begin{equation}
    \label{eq:gut4}
    \int_c^r f(t)t^{n-1} dt \geq \int_c^r g(t)t^{n-1} dt
  \end{equation}
  Since $K \subseteq \R^n$ is a convex body containing $\veczero$, for
  every line segment $\set{t\theta: 0 \leq t \leq r}$, we have that
  $\set{t\theta: 0 \leq t \leq r} \setminus K = \set{t\theta: c \leq t
    \leq r}$ for some $c > 0$. From the above inequality (equation
  \eqref{eq:gut4}), we now see that
  \begin{equation}
    \label{eq:gut5}
    \int_0^r I[t\theta \notin K] f(t)t^{n-1} dt \geq \int_0^r I[t\theta \notin K] g(t)t^{n-1} dt
  \end{equation}
  for all $\theta \in S^{n-1}$. Combining equations \eqref{eq:gut1}
  and \eqref{eq:gut2} with the above inequality \eqref{eq:gut5} yields
  the result.
\end{proof}

The following lemma establishes the necessary bounds for our ball overlap characterization of the smoothing parameter.%

\begin{lemma}
  \label{lem:ball-to-smoothing}
  Let $\lat$ denote an $n$-dimensional lattice. For $r > 0$, $0 <
  \delta < \frac{1}{4}$ and $s = r\sqrt{\frac{2\pi}{n}}$, the
  following holds:
  \[ \left( \frac{\rho_{\frac{s}{1+\delta}}(\lat \setminus
      \set{\veczero})}{\rho_{\frac{s}{1+\delta}}(\lat)} -
    e^{-\frac{2n}{3} \delta^2} \right) \leq \overlap(\lat,r) \leq
  2\rho_{2s}(\lat \setminus \set{\veczero}). \]
\end{lemma}
\begin{proof}
For simplicity of notation, we define $S = \bigcup_{\vecy \in \lat \setminus \set{\veczero}} \left(rB_2^n \cap (rB_2^n + \vecy)\right)$,
noting that $\overlap(\lat,r) = \vol_n(S)/\vol_n(rB_2^n)$. For the upper bound, by the union bound and Lemma \ref{lem:sphere-int}, we have
  \begin{align*}
    \frac{\vol_n(S)}{\vol_n(rB_2^n)} \leq
    \sum_{\vecy \in \lat \setminus \set{\veczero}} \frac{\vol_n(rB_2^n
      \cap rB_2^n + \vecy)}{\vol_n(rB_2^n)} \leq \sum_{\vecy \in \lat
      \setminus \set{\veczero}} 2e^{-\pi\length{\frac{\vecy}{2s}}^2} =
    2\rho_{2s}(\lat \setminus \set{\veczero})
  \end{align*}
  as needed.

  For the lower bound, let $\calV = \calV(\lat)$.  We claim that
  $rB_2^n \setminus \calV \subseteq S$. To see this
  take $\vecx \in rB_2^n \setminus \calV$. Since $\vecx \notin \calV$, there
  exists $\vecy \in \lat \setminus \set{\veczero}$ such that $\length{\vecy-\vecx}
  < \length{\vecx} \leq r$. Therefore $\vecx \in rB_2^n \cap \vecy + rB_2^n$ as
  needed.

  Let $s' = \frac{s}{1+\delta}$. By Lemma \ref{lem:voronoi-gaussian},
  we have that $\gamma_{s'}(\calV) \leq \frac{1}{\rho_{s'}(\lat)}$.
  Let $X \in \R^n$ denote the gaussian with distribution $\gamma$. By
  the standard gaussian tailbound (Lemma \ref{lem:gaussian-tail}), we
  have that
  \[
  \gamma_{s'}(\R^n \setminus rB_2^n) = \Pr[\length{s'X}^2 \geq r^2] =
  \Pr[\length{X}^2 \geq (1+\delta)^2 \frac{n}{2\pi}] \leq
  \Pr[\length{X}^2 \geq (1+2\delta) \frac{n}{2\pi}] \leq
  e^{-\frac{2n}{3} \delta^2}.
  \]
  Therefore we have that
  \[
  \gamma_{s'}(rB_2^n \setminus \calV) \geq 1 - \gamma_{s'}(\R^n
  \setminus \calV) - \gamma_{s'}(\R^n \setminus rB_2^n) \geq
  \frac{\rho_{s'}(\lat \setminus \set{\veczero})}{\rho_{s'}(\lat)} -
  e^{-\frac{2n}{3} \delta^2}
  \]
  By Lemma \ref{lem:gaus-unif-trans}, we have that
  \[
  \frac{\rho_{s'}(\lat \setminus \set{\veczero})}{\rho_{s'}(\lat)} -
  e^{-\frac{2n}{3} \delta^2} \leq \frac{\gamma_{s'}(rB_2^n \setminus
    \calV)}{\gamma_{s'}(rB_2^n)} \leq \frac{\vol_n(rB_2^n \setminus
    \calV)}{\vol_n(rB_2^n)} \leq \frac{\vol_n(S)}{\vol_n(rB_2^n)}
  \]
  as needed.
\end{proof}

\begin{proof}[Proof of Lemma \ref{lem:ball-overlap-char} (Ball Overlap Chacterization)]
Let $r_\eps = \frac{1}{2\eta_{\eps}(\lat^*)} \sqrt{\frac{n}{2\pi}}$ and $s = \frac{1}{\eta_{\eps}(\lat^*)}$. Note that by definition of $s$, we have that
$\rho_s(\lat) \leq \eps$. For Part $1$, by Lemma \ref{lem:ball-to-smoothing}, using the fact that $s = 2r_\eps \sqrt{\frac{2\pi}{n}}$, we get
\[
\overlap(\lat, r_\eps) \leq 2\rho_s(\lat) \leq 2\eps.
\]
Furthermore, for every $0<r \leq r_{\eps}$, since $\eta_\eps(\lat)$ is a monotonically decreasing function in $\eps$, $r = r_{\eps'}$ for some $\eps' \leq \eps$. Thus,
\[
\overlap(\lat, r) \leq 2\eps' \leq 2\eps.
\]
For Part 2, by \ref{lem:ball-to-smoothing} and the fact that $2^{-o(n)} \leq \eps \leq 1/3$ and $\delta = \sqrt{\frac{3}{2n} \ln \frac{4}{\eps}} < 1/4$, we get
\[
\overlap(\lat, 2\cdot (1+\delta)\cdot r_\eps) \geq \frac{\rho_{s}(\lat^* \setminus \set{\veczero})}{\rho_{s}(\lat^*)}- e^{-\frac{2n}{3}\delta^2} \geq \frac{\eps}{1+\eps} - \frac{\eps}{4} \geq \frac{\eps}{2}.
\]
Again, for every $r \geq 2\cdot (1+\delta) r_{\eps}$, by the monotonicity of $\eta_\eps(\lat)$ in $\eps$, $r = r_{\eps'}$ for some $\eps' \geq \eps$, and thus,
\[
\overlap(\lat, r) \geq \eps'/2 \geq \eps/2.
\]
\end{proof}

\subsection{Background and From $\ID$ Commitment Schemes to $\SZK$ Protocols} \label{subsec:overview-SZK}

An $\ID$ commitment scheme $\Com$ for a promise problem $\Pi$ is a commitment scheme that can depend on the instance $x$ and such that only one of the hiding and binding properties are required to hold, depending on whether $x$ is an YES or NO instance. Since only one of the hiding and binding properties needs to hold at a time, it is possible to achieve both the statistical hiding and statistical binding properties, and thus useful for constructing $\SZK$ protocols.

Typically, one requires the hiding property to hold for the YES instances and the binding property to hold for the NO instances, and such an $\ID$ commitment scheme readily gives a $\SZK$ protocol with soundness error $1/2$.  On the other hand, an $\ID$ commitment scheme with reverse guarantees, i.e., binding for YES instances and hiding for NO instances, also readily gives a honest verifier $\SZK$ protocol, where the verifier commits to a random bit $b$ and the prover's task is to guess the bit $b$ correctly. Furthermore, since the verifier (who is the sender of the $\ID$ commitment scheme) is honest, the binding property only needs to hold with respect to the honest sender (referred to as ``honest-sender biding property''). Since $\HVSZK = \SZK$~\cite{GoldreichSV98}, an $\ID$ commitment scheme that is honest-sender binding for YES instances and hiding for NO instance is also sufficient for showing that the promise problem is in $\SZK$. 
Note that since only honest-sender binding property is required, we can without loss of generality assume that a commitment scheme is non-interactive (by letting the sender emulate the receiver and send the emulated view to the receiver). Thus, such a commitment scheme is simply an algorithm. %

We observe that, the existing security amplification techniques for regular commitment schemes can be applied to the instance-dependent setting. As a consequence, any $\ID$ commitment scheme with ``\emph{non-trivial}'' honest-sender binding and hiding properties is sufficient to obtain $\SZK$ protocols. More precisely, as formally defined in Definition~\ref{def:ID-Com}, we consider $\ID$ commitment schemes $\Com$ with weak $p$-hiding and $q$-binding properties, where the hiding and binding properties can be broken with ``advantage'' at most $p$ and $q$, respectively, and we say $\Com$ is ``non-trivial'' if $p+q \leq 1 - 1/\poly(n)$. Known security amplification results for commitment schemes (for the case of statistical security)~\cite{DamgardKS99} state that any non-trivial commitment scheme can be amplified to one with full-fledge security (i.e., both $p$ and $q$ are negligible). The same conclusion holds for $\ID$ commitment schemes, and thus to construct a $\SZK$ protocol for a language $L$, it suffice to construct a non-trivial honest-sender binding $\ID$ commitment scheme for $L$.

\begin{theorem} \label{thm:Com-amplification} Let $\Pi$ be a promise problem. Suppose there exists a \emph{non-trivial} $\ID$ commitment scheme for $\Pi$, then $\Pi \in \SZK$.
\end{theorem}
\begin{proof}(sketch)
The theorem can be proved by applying known technique/results for regular commitment schemes to the instance-dependent setting. Briefly, security amplification of commitment schemes can be done using the following two operations~\cite{DamgardKS99}.
\begin{itemize}
\item {\bf Repetition.} Given $\Com$ and $k\in \mathbb{N}$, define $\Com'_x(b) = (\Com_x(b;r_1),\dots,\Com_x(b;r_k))$, i.e., concatenation of $k$ commitments of $\Com$ using independent randomness. This amplifies the binding property but degrades the hiding property. Specifically, if $\Com$ is $p$-hiding and $q$-binding, then $\Com'$ is $(1-(1-p)^k)$-hiding and $q^k$-binding.
\item {\bf Sharing.} Given $\Com$ and $k\in \mathbb{N}$, define $\Com'_x(b) = (\Com_x(b_1;r_1),\dots,\Com_x(b_k;r_k))$, where $b_1,\dots,b_k$ are chosen randomly subject to $b_1 \oplus \cdots \oplus b_k = b$, and $r_1,\dots,r_k$ are independent randomness. This amplifies the hiding property but degrades the binding property. Specifically, if $\Com$ is $p$-hiding and $q$-binding, then $\Com'$ is $p^k$-hiding and $1- (1-q)^k$-binding.
\end{itemize}
It can be shown (as in~\cite{DamgardKS99}) that as long as $p + q \leq
1 - 1/\poly(n)$, one can amplify a $p$-hiding and $q$-binding
commitment scheme $\Com$ to a secure $\Com'$ by alternately applying
repetition and sharing operations with carefully chosen parameters
$k$'s, and the resulting $\Com'$ calls $\Com$ in a black-box way
$\poly(n)$ times.

Once we have a secure non-interactive instance-dependent bit-commitment scheme for $\Pi$, we can readily construct a two-message honest-verifier $\SZK$ protocol for $L$ as follows: On input $x \in \zo^n$,
\begin{itemize}
\item $V$ samples random $b \leftarrow \zo$, computes and sends $\Com_x(b)$ to $P$.
\item $P$ sends $b'$ to $V$ as his guess of $b$.
\item $V$ accepts iff $b' = b$.
\end{itemize}
It is not hard to see that the binding and hiding properties translate to the completeness and $1/2$-soundness for the protocol, and a simulator can
generate the view by emulating $V$ and outputting $(\Com_x(b), b)$. Since $\HVSZK = \SZK$, we have $\Pi \in \SZK$.
\end{proof}

\begin{remark} Interestingly, as a by-product, an $\SZK$-complete problem called ``Image Intersection Density'' (IID) (defined by ~\cite{BG03} and proved to be $\SZK$-complete by~\cite{CCKV08}) can naturally be interpreted as a weak $\ID$ bit-commitment scheme as defined in Definition~\ref{def:ID-Com}, which allows us to (immediately) obtain an optimal ``polarization'' result to the problem.
 
Specifically, the input to the IID problem is two distributions $(X,Y)$ specified by circuits, where the YES instance satisfying $\Delta(X,Y) \leq a$ and the NO instance satisfying $\Pr[ X \notin \supp(Y)] \geq b$ and $\Pr[ Y \notin \supp(X)] \geq b$, where $a,b \in (0,1)$ are parameters of the problem. By defining $X$ and $Y$ as commitment to $0$ and $1$ respectively, the condition to YES instance corresponds to statistical $a$-hiding and the condition to NO instance corresponds to statistical honest-sender $(1-b)$-binding.\footnote{The binding and hiding properties hold for reverse instances, but one can instead consider the complement of the IID problem to obtain a consistent definition since $\SZK$ is close under complement.}
Interpreting the IID problem as a weak $\ID$ bit-commitment scheme makes it natural to apply the security amplification result of commitment
schemes~\cite{DamgardKS99}, which gives an optimal polarization result of the problem, stating that the IID problem with parameters $a(n) - b(n) \geq 1/\poly(n)$ is complete for $\SZK$. This improves the previous known result in~\cite{CCKV08}, which holds for constants $a>b$. In fact, the security amplification and polarization techniques exploit identical operations. The stronger result from the security amplification literature is obtained by applying the repetition and sharing operations more carefully.
\end{remark}

\section{Applications to Worst-case to Average-case Reductions} \label{sec:wstavg}

Our study of $\gapsp$ has natural applications to the context of worst-case to average case reductions. In particular, we show that we can relate the hardness of average-case hard learning with error (\LWE) problems and worst-case hard $\gapsp$ problems with a tighter connection factor. Our result directly implies the worst-case to average-case result from $\gapsvp$ to $\LWE$ obtained by Regev~\cite{DBLP:journals/jacm/Regev09} and Peikert~\cite{DBLP:conf/stoc/Peikert09}.  First we review the $\LWE$ problem.

\begin{definition}[Learning with Error Problem~\cite{DBLP:journals/jacm/Regev09}] Let $q=q(n)\in \N$, $\alpha= \alpha(n) \in (0,1)$.
Let $\Phi_{\alpha}$ be the distribution on $[0,1)$ obtained by drawing a sample from the Gaussian distribution with standard deviation $\alpha$ and reducing it modulo $1$.
Define $A_{\vecs,\Phi_{\alpha}}$ to be the distribution on $\Z^{n}_{q} \times [0,1)$ obtained by choosing a vector $\veca  \in \Z^{n}_{q}$ uniformly at random, choosing an error term $e \leftarrow \Phi_{\alpha}$, and outputting $(\veca, \langle \veca , \vecs \rangle / q + e )$ where the addition is performed in modulo $1$.

The goal of the learning with errors problem $\LWE_{q,\alpha}$ in $n$ dimensions is, given access to any desired $\poly(n)$ numbers of samples from $A_{\vecs, \Phi_{\alpha}}$ for a random $\vecs \leftarrow \Z_{q}^{n}$, to find $\vecs$ (with overwhelming probability).

\end{definition}

Following~\cite{DBLP:journals/jacm/Regev09,DBLP:conf/stoc/Peikert09}, we use the bounded decoding  $\BDD$ problem as an intermediate step in our reduction. Here we instead parameterize the $\alpha$-$\BDD$ problem with $\alpha$ relative to the smoothing parameter (as opposed to the shortest vector used in literature); this is essential for us to obtain tighter reduction for $\gapspp$.

\begin{definition}[Bounded Distance Decoding Problem ($\alpha$-$\BDD_{\eps}$)] Given a lattice basis $B$ and a vector $\vect$ such that $\dist (\vect, \lat(B)) < \alpha /\eta_{\eps}(\lat(B)^{*})$, find the lattice vector $\vecv \in \lat(B)$ such that  $\dist (\vect, \vecv) \leq \alpha /\eta_{\eps}(\lat(B)^{*})$.

\end{definition}

We recall the following Lemma from Regev~\cite{DBLP:journals/jacm/Regev09} and Peikert~\cite{DBLP:conf/stoc/Peikert09} that reduce solving worst-case $\BDD$ problem to solving $\LWE$ through quantum and classic reductions, respectively. %

\begin{lemma}[\cite{DBLP:journals/jacm/Regev09,DBLP:conf/stoc/Peikert09}]
Let $q(n) \in \N$, $\alpha(n) \in (0,1)$, $\eps(n)$ be a negligible function such that $\alpha \cdot q > 2\sqrt{n}$. There exists a $\PPT$ quantum reduction from solving  $\alpha/2$-$\BDD_{\eps}$ in the worst case (with overwhelming probability) to solving $\LWE_{q,\alpha}$ using $\poly(n)$ samples. 

If in addition $q\geq 2^{n/2}$, then there exists a classical reduction  from solving  $\alpha/2$-$\BDD_{\eps}$ in the worst case (with overwhelming probability) to solving $\LWE_{q,\alpha}$ using $\poly(n)$ samples.
\end{lemma}

We note that the reason $\eps = \negl(n)$ in the above Lemma is to guarantee that the LWE samples generated during the reduction are within neglible
statistical distance from ``true'' LWE samples.

We now establish a new result that relates $\BDD$ and $\gapsp$. Our new observation is that the prover in the GGG protocol
(Algorithm~\ref{am-protocol}) can be implemented by a  $\BDD$ oracle.  Thus, if one has a $\BDD$ solver, one can solve the $\gapsp$ problem. We note
that we only need the $\BDD$ oracle to work for YES instances, and hence we require $\eps_Y = \negl(n)$ while leaving $\eps_N=\frac{1}{\poly(n)}$.
More precisely, we have the following lemma.

\begin{lemma} 
Let $\alpha(n) \in (0,1)$, $\eps_{Y}(n) \in \negl(n)$ and $\eps_{N} \in 1/ \poly(n)$. There exists a $\PPT$ Turing reduction from solving $\sqrt{n}/ \alpha$- $\gapsp_{\eps_{Y},\eps_{N}}$ to solving  $\alpha$-$\BDD_{\eps_{Y}}$.
\end{lemma}
\begin{proof} For convenience, we
scale the $\sqrt{n}/ \alpha$-$\gapsp$ problem so that YES instances have $\eta_{\epsilon_{Y}}(\lat) \leq \alpha/\sqrt{n}$, and  NO instances have
$\eta_{\epsilon_{N}}(\lat) > 1$. 
Let $B$ be an input of the problem $\sqrt{n}/ \alpha$-$\gapsp$. We run the GGG protocol as Algorithm~\ref{am-protocol} on input $B$, where the prover's strategy is implemented using the  $\alpha$-$\BDD_{\eps_{Y}}$ solver. Then we output the verifier's decision. 

Now we describe the  analysis. For
 NO instances, by an identical analysis to Theorem~\ref{thm:GGG}, 
 the above algorithm rejects with probability at least $\eps_{N}/(1+\eps_{N}) > 1/\poly(n)$. For YES instances, we observe that the optimal prover's strategy can be  emulated if $\|\vecx\| $ is less than the BDD decoding distance $\alpha/\eta_{\eps_{Y}}(\lat) \geq \sqrt{n}$. By the Gaussian tail bound as Lemma~\ref{lem:gaussian-tail}, we have $\Pr[\|\vecx\| \geq \sqrt{n}] < e^{-\Omega(n)}$. Recall that by  Lemma~\ref{lem:voronoi-gaussian}, in GGG protocol the verifier rejects the optimal prover with probability $1-\gamma_{1}(\calV(\lat^{*})) \leq \rho_{2}(\lat^{*} \setminus \veczero) \leq \eps_{Y}$.
 Thus, by a union bound  the algorithm rejects with probability at most 
  $\eps_{Y} +  e^{-\Omega(n)} \leq \negl(n) $.

\end{proof}

Putting together the above lemmas, we obtain a tighter worse-case to average-case reduction from $\gapspp$ to $\LWE$.
\begin{theorem}
Let $q(n) \in \N$, $\alpha(n) \in (0,1)$, $\eps_{Y}(n) \in \negl(n)$ and $\eps_{N} \in 1/ \poly(n)$ such that $\alpha \cdot q > 2\sqrt{n}$.
 There exists a $\PPT$ quantum reduction from solving  $2\sqrt{n}/ \alpha$- $\gapsp_{\eps_{Y},\eps_{N}}$ in the worst case (with overwhelming probability)  to solving $\LWE_{q,\alpha}$ using $\poly(n)$ samples.
 
If in addition $q\geq 2^{n/2}$, then there exists a classical reduction  from solving $2\sqrt{n}/ \alpha$- $\gapsp_{\eps_{Y},\eps_{N}}$ in the worst case (with overwhelming probability)  to solving $\LWE_{q,\alpha}$ using $\poly(n)$ samples.
 
\label{thm:gapsp-to-lwe}
\end{theorem}

\begin{remark} By using the following relation of shortest vectors and smoothing parameters by Micciancio and Regev~\cite{DBLP:journals/siamcomp/MicciancioR07}:
 $$\frac{\sqrt{\log(1/\eps)}}{\sqrt{\pi}\lambda_{1}(\lat^{*})} \leq \eta_{\eps}(\lat) \leq \frac{\sqrt{n}}{\lambda_{1}(\lat^{*})} \mbox{ for } \eps \in [2^{-n},1], $$
 the above theorem implies that  there exists a  corresponding $\PPT$ quantum/classical reduction from $(c\cdot \frac{n}{\alpha \sqrt{\log n}})$-$\gapsvp$ to $\LWE_{q,\alpha}$ for any constant $c > 0$.

\end{remark}

\section{Co-AM Protocol for \gapsp}
\label{sec:coAM}

In this section, we describe an co-AM protocol for $\gapsp$.
Formally, we establish the following:

\begin{theorem}
  \label{thm:coam}
  For any $\alpha \geq 1/\poly(n)$ and $\eps_{Y},\eps_{N}$ such that
  $\eps_{N} \geq (1+1/\poly(n)) \cdot \eps_{Y}$,
  we have
  $(1+\alpha)$-$\gapsp_{\eps_{Y},\eps_{N}} \in \coAM$.
\end{theorem}

By applying Corollary~\ref{cor:eps-epsY-epsN}, we obtain the following  upper bound on the complexity of $\gamma$-$\gapspp_{\eps}$. 

\begin{corollary}
  For every $\eps:\N \rightarrow (0,1)$  such that  $\eps(n) < 1- 1/\poly(n)$, we have $(1+o(1))$-$\gapspp_{\eps} \in \coAM$. 

\end{corollary}
\fnote{need to check Corollary~\ref{cor:eps-epsY-epsN} to see whether we can apply $\eps >1$}

Our main tool is the classic set size lower bound protocol by
Goldwasser and Sipser~\cite{GoldwasserS86}.  We use this protocol to
show that the smoothing parameter should be at least as large as some
quantity. To show that $\eta(\lat)$ is large, equivalently we are showing 
that the discrete Gaussian weights 
are large for the points in $\lat^{*}$ inside the $\sqrt{n}$ ball\footnote{Actually the radius needs to depend on the parameter $\eps_{Y}$. Here for simplicity we think $\eps_{Y}$ as a constant.}. (The Gaussian weights outside 
the ball becomes exponentially small.)

The set size 
lower bound protocol gives a very accurate approximation of lattice points inside the
$\sqrt{n}$ ball, but   its set size is not sufficient to approximate the Gaussian weights. 
The two points inside the ball could have lengths that differ a lot, and thus 
their Gaussian weights differ even more. Our new observation is that we 
can partition the $\sqrt{n}$ ball into different shells (con-centered at $\veczero$),
and then use the set size protocol to approximate the number of lattice points 
lying in each shell.
Since every point in the same shell has roughly the same length and thus 
Gaussian weight,
we can approximate the total Gaussian weights in a shell according to 
the size. Thus, summing up the Gaussian weight of each shell, 
we are able to approximate the Gaussian weights inside the $\sqrt{n}$ ball.
Thus, we are able to show that the Gaussian weights inside the ball are large, and thus 
$\eta$ is large.

 First we describe the set size lower bound protocol:

\begin{definition}[Set size lower bound protocol~\cite{GoldwasserS86}]
  Let $\Ver$ be a probabilistic polynomial time verifier, and $\Prov$
  be a (computationally unbounded) prover. Let $S \subseteq \zo^{n}$
  be a set whose membership can be efficiently certified. The two
  parties hold common inputs $1^{n}$ and $K \in \N$.

  We say $\langle \Prov, \Ver \rangle $ is a
  $(1-\gamma)$-approximation protocol of the set size $|S|$ if the
  following conditions hold:

  \begin{itemize}
  \item (Completeness) If $|S| \geq K$, then $\Ver$ will always
    accept. %
  \item (Soundness) If $|S| < (1-\gamma) \cdot K$, then $\Ver$ will
    accept with probability at most $\negl(n)$ for some negligible
    function $\negl(\cdot)$.
  \end{itemize}
\end{definition}

\noindent Now we recall the classic construction of the set size lower
bound protocol:

\begin{theorem}[\cite{GoldwasserS86}]
  For any set $S \in \zo^{n}$ whose membership can be efficiently
  certified, and any $\gamma = 1/\poly(n)$, there exists a
  public-coin, 2-round $(1-\gamma)$-approximation protocol of the set
  size $|S|$.

  Moreover, for any $k=\poly(n)$, we can run the protocol $k$-times in
  parallel for $k$ set-number pairs $\{(S_{i},K_{i})\}_{i\in [k]}$,
  and the resulting protocol has perfect completeness and negligible
  soundness error. Here soundness error means the probability that
  there exists some $i^{*} \in [k]$ such that $|S_{i} |\leq (1-\gamma)
  \cdot K_{i} $ but $\Ver$ accepts.
\end{theorem}

\begin{proof}[Proof of Theorem~\ref{thm:coam}]
  To show the theorem, we first describe a $\coAM$ protocol $\langle
  \Prov, \Ver \rangle$ in the following. Note that the verifier in a
  $\coAM$ protocol must accept the NO instances and reject the YES
  instances of $(1+\alpha)$-$\gapsp_{\eps_{Y},\eps_{N}}$.  For
  convenience, the YES or NO instances here are with respect to the
  $\gapsp$ problem, so the completeness means the verifier accepts any
  NO instance, and the soundness means he rejects any YES instance.

  Let $B$ be an $n$-dimensional basis of a lattice $\lat$ as input to
  the prover and verifier, satisfying either $\eta_{\eps_{N}}(\lat)
  \geq (1+ \alpha)$ (NO instance) or $\eta_{\eps_{Y}} (\lat) \leq 1$
  (YES instance), where $\alpha \geq 1/\poly(n)$, $\eps_{N} \geq (1+1/\poly(n)) \cdot \eps_{Y}$. 
   The prover and the verifier agree on the
  following parameters:

  \paragraph{Parameters.} Let $R= n \cdot (1+\log(1/\eps_{Y}))$,
 $1-\beta =  \frac{2\eps_{Y}}{\eps_{Y}+\eps_{N}}$,
  and let $T= \lceil \frac{ \log \sqrt R}{\log(1+\alpha)}
  \rceil$. We know for $\alpha \geq 1/\poly(n)$ being noticeable, we
  have $T$ bounded by some polynomial, i.e. $T\leq \poly(n)$.  Then we
  define spaces $S_{0}\eqdef \left\{ \vecv\in \lat^{*}: 0<
    \length{\vecv} \leq 1 \right\}$, and $S_{i} \eqdef \left\{ \vecv
    \in \lat^{*}: (1+\alpha)^{i-1}< \length{\vecv} \leq (1+\alpha)^{i}
  \right\}$, for $i \in [T]$.  Pictorially, these $S_{i}$'s form a
  partition of space inside the region of $\sqrt{R} B^{n}_{2}$. Each
  $S_{i}$ is a shell that contains lattice points from length
  $(1+\alpha)^{i-1}$ to $(1+\alpha)^{i}$.
 
  Then $\langle \Prov, \Ver \rangle $ does the following:

  \begin{itemize}

  \item $\Prov$ sends $K_{0},K_{1},K_{2}, \dots K_{T} \in \N$ as
    claims of the sizes of $S_{0},S_{1}, S_{2}, \dots, S_{T} $.
 
  \item Then for each pair $(S_{i},K_{i})$, $\Prov$ and $\Ver$ run the
    $(1-\beta)$-approximation protocol as its subroutine.  These $T$
    approximation protocols are run in parallel. %
   Note that $\frac{2\eps_{Y}}{\eps_{Y}+\eps_{N}} \leq 1-1/\poly(n)$ since $\eps_{N} \geq (1+1/\poly(n)) \cdot \eps_{Y}$. Thus, $1-\beta \leq 1-1/\poly(n)$, which is within the range of parameters of the set size lower bound. 
  
  \item In the end, $\Ver$ accepts if and only if all the
    approximation subprotocols are accepted, and $\sum_{0\leq i \leq
      T} K_{i} \cdot e^{-\pi (1+\alpha)^{2i} } \geq (\eps_{Y} + \eps_{N})
    /2$. 
  \end{itemize}

  It is easy to see that the verifier can be implemented in
  probabilistic polynomial time. It remains to show the completeness
  and soundness. We show them by the following two claims:

\begin{claim}
  If $B$ is a NO instance, $\Ver$ will always accept the honest
  prover's strategy.
\end{claim}

\begin{proof}
  Let $K_{0},K_{1},\dots, K_{T}$ be the values of the set sizes
  $S_{0},S_{1},\dots, S_{T}$, as the honest prover will always send
  the correct values. From the promise of NO instances, we know
  $\eta_{ \eps_{N}}(\lat) \geq (1+ \alpha) $, which implies
  \[ q \eqdef \sum_{\vecv\in \lat^{*} \setminus \{\veczero\}} e^{-\pi
    (1+\alpha)^{2} \length{\vecv}^{2}} \geq \eps_{N}.\]

  By rearranging the order of summation, we have
  \begin{eqnarray*}
    \lefteqn{ q } 
    & = & \sum_{0\leq i \leq T} \sum_{\vecv\in S_{i}} e^{-\pi (1+\alpha)^{2}  \length{\vecv}^{2}}  + \sum_{v\in \lat^{*} \setminus (\sqrt{R} \cdot B^{n}_{2})} e^{-\pi (1+\alpha)^{2}  \length{\vecv}^{2}}  \\
    &\leq& \sum_{0\leq i \leq T} \sum_{\vecv\in S_{i}} e^{-\pi (1+\alpha)^{2}  \length{\vecv}^{2}}   +  2^{-n}\cdot \eps_{Y}  \\
    &\leq &  \sum_{0\leq i \leq T} K_{i} \cdot e^{-\pi  (1+\alpha)^{2i} } + 2^{-n} \cdot \eps_{Y}.
  \end{eqnarray*}
  The first equality comes from the rearrangement; the second line is
  a tail bound inequality by Lemma \ref{lem:gaussian-tail} by plugging
  suitable parameters; the last inequality is by the fact that $\vecv \in
  S_{i}$ implies $\length{\vecv}\geq (1+\alpha)^{i-1}$ for $i\in [T]$.
 
  Then we have
 \[ \sum_{0\leq i \leq T} K_{i} \cdot e^{-\pi  (1+\alpha)^{2i} }  \geq
 \eps_{N} -2^{-n} \cdot \eps_{Y} \geq (\eps_{Y} + \eps_{N})/2, \]
 for all sufficiently large $n$'s.  This follows by the fact that
 $\eps_{N}\geq (1+1/\poly(n)) \cdot \eps_{Y}$,
 and a straightforward examination.  
 Thus, the verifier will always accept.
\end{proof}

\begin{claim}
  If $B$ is a YES instance, then no prover can convince the verifier
  with probability better than a negligible quantity.
\end{claim}

\begin{proof}
  From the promise of YES instances, we know $\eta_{ \eps_{Y}}(\lat)
  \leq 1$, which implies
  \[ q \eqdef \sum_{\vecv\in \lat^{*} \setminus \{ \veczero\}} e^{-\pi
    \length{v}^{2}} \leq \eps_{Y}.
  \]
  Similarly, we rearrange the order of summation and get
  \begin{eqnarray*}
    \lefteqn{ q } 
    & = & \sum_{0\leq i \leq T} \sum_{\vecv\in S_{i}} e^{-\pi \length{\vecv}^{2}}  + \sum_{\vecv\in \lat^{*} \setminus (\sqrt{R} \cdot B^{n}_{2})} e^{-\pi   \length{\vecv}^{2}}  \\
    &\geq& \sum_{0\leq i \leq T} \sum_{\vecv\in S_{i}} e^{-\pi   \length{\vecv}^{2}}  \\
    &\geq & \sum_{0\leq i \leq T} |S_{i}| \cdot e^{-\pi (1+\alpha)^{2i} }.
  \end{eqnarray*}

  Suppose the prover sends some $K_{0},K_{1} \dots K_{T}$ such that
  $\sum_{0\leq i \leq T} K_{i} \cdot e^{-\pi (1+\alpha)^{2i} } \geq
  (\eps_{Y} + \eps_{N})/2$, it must be the case that $\exists i^{*} \in
  [T]$ such that $\frac{\eps_{Y}}{(\eps_{Y}+ \eps_{N})/2} K_{i}  = (1-\beta) \cdot K_{i}\geq
  |S_{i}|$ from a simple counting argument.
   By the soundness of the
  $(1-\beta)$-approximation protocol, the verifier will catch this
  with probability $(1-\negl(n))$. Hence the verifier accepts a YES instance with only negligible probability.
\end{proof}

Together with the two claims, the proof of the theorem is complete.
\end{proof}

\section{Deterministic Algorithm for Smoothing Parameter}

In this section we show that $(1+o(1))$-$\gapsp$ can be solved deterministically
in time $2^{O(n)}$. In particular we are able to show the following theorem.

To show the theorem, use are going to establish the following lemma.

\begin{theorem}
\label{thm:det}
For any $\eps_{Y},\eps_{N}:\N \rightarrow [0,1]$ such that  $\eps_{N}(n)-\eps_{Y}(n) \geq 1/2^{-2n}$, $1$-$\gapsp_{\eps_{Y},\eps_{N}} \in \mathit{DTIME}(2^{O(n)} )$.
\end{theorem}

Together with Corollary \ref{cor:eps-epsY-epsN}, we are able to obtain the following corollary.
\begin{corollary}
\label{col:det}
For any $\eps: \N\rightarrow [0,1]$ and $ \eps(n) \geq  2^{-n}$,
the problem $(1+o(1))$-$\gapsp_{\eps} \in \mathit{DTIME} (2^{O(n)})$.
\end{corollary}

We will crucially use the following lattice point enumeration
algorithm. The algorithm is a slight tweak of closest vector problem
algorithm of Micciancio and Voulgaris
\cite{DBLP:conf/stoc/MicciancioV10}, which was first used in
\cite{DBLP:conf/focs/DadushPV11} to solve the shortest vector problem
in general norms.

\begin{proposition}[\cite{DBLP:conf/stoc/MicciancioV10,DBLP:conf/focs/DadushPV11}, Algorithm Ball-Enum]
There is an algorithm Ball-Enum that given a radius $r > 0$, a basis
$B$ of an $n$-dimensional lattice $\lat$, and $t \in \R^n$, lazily
enumerates the set $\lat \cap (rB_2^n +t)$ in deterministic time
$2^{O(n)} \cdot (|L \cap (t+rB_2^n)|+1)$ using at most $2^{O(n)}$ space.
\end{proposition}

Now we are ready to prove Theorem \ref{thm:det} using the above theorem.
\begin{proof} Let $B$ be an $n$-dimensional basis of a lattice $\lat$ satisfying either $\eta_{\eps_{Y}} \leq 1$ or $\eta_{\eps_{N}} \geq 1$, where $\eps_{Y},\eps_{N}$ are parameters that  the conditions in the theorem hold. Now we are going to describe an algorithm $A$ on input $B$ that distinguishes the two cases. 

 $A$ runs the enumeration algorithm with the parameters $t=0$, $r=\sqrt{n}$ to enumerate all points in $\lat^{*} \cap \sqrt{n}\cdot B^{n}_{2}$. If $A$ has already found $e^{\pi \cdot n} \cdot \eps_{N}$ points from the enumeration algorithm, $A$ terminates and rejects immediately. This is because 
 
 $$
 \sum_{\vecv \in \lat^{*} \setminus \{0\}}  e^{-\pi \|\vecv\|^{2}} \geq \sum_{\vecv \in ( \lat^{*} \setminus \{0\})  \cap \sqrt{n} \cdot B^{n}_{2}}  e^{-\pi \|\vecv\|^{2}} \geq ( e^{\pi n} \cdot  \eps_{N} ) \cdot e^{-\pi   n}  =\eps_{N},
  $$
 which already implies the case of no instances.
 
 Otherwise $A$ computes $u = \sum_{\vecv \in (\lat^{*} \setminus \{0\}) \cap \sqrt{n}\cdot B^{n}_{2}  } e^{-\pi  \|\vecv\|^{2}}$. $A$ accepts if $u \leq (\eps_{Y} + \eps_{N})/2$, and otherwise rejects. The analysis of its completeness and soundness is very similar to  that of Theorem \ref{thm:coam}, so we do not restate it here. It is not hard to see that this can be done in time $ 2^{O(n)} \cdot  e^{\pi  n} \cdot  \eps_{N} = 2^{O(n)} $.

\end{proof}

%
%
%
%

%

\section*{Acknowledgments}
We thank Oded Regev for helpful discussions. We also thank the anonymous reviewers for helpful suggestions.

\bibliographystyle{alphaabbrvprelim}

\begin{thebibliography}{CCKV08}
\expandafter\ifx\csname urlstyle\endcsname\relax
  \providecommand{\doi}[1]{doi:\discretionary{}{}{}#1}\else
  \providecommand{\doi}{doi:\discretionary{}{}{}\begingroup
  \urlstyle{rm}\Url}\fi

\bibitem[Ajt96]{ajtai04:_gener_hard_instan_lattic_probl}
M.~Ajtai.
\newblock Generating hard instances of lattice problems.
\newblock \emph{Quaderni di Matematica}, 13:1--32, 2004.
\newblock Preliminary version in STOC 1996.

\bibitem[AR04]{DBLP:journals/jacm/AharonovR05}
D.~Aharonov and O.~Regev.
\newblock Lattice problems in {NP} $\cap$ {coNP}.
\newblock \emph{J. ACM}, 52(5):749--765, 2005.
\newblock Preliminary version in FOCS 2004.

\bibitem[Bal97]{Ball97}
K.~M. Ball.
\newblock An elementary introduction to modern convex geometry.
\newblock \emph{In S. Levy (Ed.), Flavors of Geometry, Number 31 in MSRI
  Publications}, pages 1--58, 1997.

\bibitem[Ban93]{banaszczyk93:_new}
W.~Banaszczyk.
\newblock New bounds in some transference theorems in the geometry of numbers.
\newblock \emph{Mathematische Annalen}, 296(4):625--635, 1993.

\bibitem[Ban95]{DBLP:journals/dcg/Banaszczyk95}
W.~Banaszczyk.
\newblock Inequalites for convex bodies and polar reciprocal lattices in
  {$R^n$}.
\newblock \emph{Discrete {\&} Computational Geometry}, 13:217--231, 1995.

\bibitem[Ban96]{Bana96}
W.~Banaszczyk.
\newblock Inequalities for convex bodies and polar reciprocal lattices in
  {$R^n$} {II}: Application of k-convexity.
\newblock \emph{Discrete and Computational Geometry}, 16:305--311, 1996.
\newblock ISSN 0179-5376.

\bibitem[BHZ87]{DBLP:journals/ipl/BoppanaHZ87}
R.~B. Boppana, J.~H{\aa}stad, and S.~Zachos.
\newblock Does co-{NP} have short interactive proofs?
\newblock \emph{Inf. Process. Lett.}, 25(2):127--132, 1987.

\bibitem[BOG03]{BG03}
M.~Ben-Or and D.~Gutfreund.
\newblock Trading help for interaction in statistical zero-knowledge proofs.
\newblock \emph{J. Cryptology}, 16(2):95--116, 2003.

\bibitem[BV11]{DBLP:conf/focs/BrakerskiV11}
Z.~Brakerski and V.~Vaikuntanathan.
\newblock Efficient fully homomorphic encryption from (standard) {LWE}.
\newblock In \emph{FOCS}, pages 97--106. 2011.

\bibitem[CCKV08]{CCKV08}
A.~Chailloux, D.~F. Ciocan, I.~Kerenidis, and S.~P. Vadhan.
\newblock Interactive and noninteractive zero knowledge are equivalent in the
  help model.
\newblock In \emph{TCC}, pages 501--534. 2008.

\bibitem[CH11]{DBLP:conf/innovations/CohnH11}
H.~Cohn and N.~Heninger.
\newblock Ideal forms of {Coppersmith's} theorem and {Guruswami}-{Sudan} list
  decoding.
\newblock In \emph{ICS}, pages 298--308. 2011.

\bibitem[CHKP10]{DBLP:conf/eurocrypt/CashHKP10}
D.~Cash, D.~Hofheinz, E.~Kiltz, and C.~Peikert.
\newblock Bonsai trees, or how to delegate a lattice basis.
\newblock In \emph{EUROCRYPT}, pages 523--552. 2010.

\bibitem[Cop97]{DBLP:journals/joc/Coppersmith97}
D.~Coppersmith.
\newblock Small solutions to polynomial equations, and low exponent {RSA}
  vulnerabilities.
\newblock \emph{J. Cryptology}, 10(4):233--260, 1997.

\bibitem[DKS99]{DamgardKS99}
I.~Damg{\aa}rd, J.~Kilian, and L.~Salvail.
\newblock On the (im)possibility of basing oblivious transfer and bit
  commitment on weakened security assumptions.
\newblock In \emph{EUROCRYPT}, pages 56--73. 1999.

\bibitem[DPV11]{DBLP:conf/focs/DadushPV11}
D.~Dadush, C.~Peikert, and S.~Vempala.
\newblock Enumerative lattice algorithms in any norm via {M}-ellipsoid
  coverings.
\newblock In \emph{FOCS}, pages 580--589. 2011.

\bibitem[Gen09]{DBLP:conf/stoc/Gentry09}
C.~Gentry.
\newblock Fully homomorphic encryption using ideal lattices.
\newblock In \emph{STOC}, pages 169--178. 2009.

\bibitem[Gen10]{DBLP:conf/crypto/Gentry10}
C.~Gentry.
\newblock Toward basing fully homomorphic encryption on worst-case hardness.
\newblock In \emph{CRYPTO}, pages 116--137. 2010.

\bibitem[GG98]{DBLP:journals/jcss/GoldreichG00}
O.~Goldreich and S.~Goldwasser.
\newblock On the limits of nonapproximability of lattice problems.
\newblock \emph{J. Comput. Syst. Sci.}, 60(3):540--563, 2000.
\newblock Preliminary version in STOC 1998.

\bibitem[GMR04]{DBLP:journals/cc/GuruswamiMR05}
V.~Guruswami, D.~Micciancio, and O.~Regev.
\newblock The complexity of the covering radius problem.
\newblock \emph{Computational Complexity}, 14(2):90--121, 2005.
\newblock Preliminary version in CCC 2004.

\bibitem[GPV08]{DBLP:conf/stoc/GentryPV08}
C.~Gentry, C.~Peikert, and V.~Vaikuntanathan.
\newblock Trapdoors for hard lattices and new cryptographic constructions.
\newblock In \emph{STOC}, pages 197--206. 2008.

\bibitem[GS86]{GoldwasserS86}
S.~Goldwasser and M.~Sipser.
\newblock Private coins versus public coins in interactive proof systems.
\newblock In \emph{STOC}, pages 59--68. 1986.

\bibitem[GS88]{DBLP:journals/siamcomp/GrollmannS88}
J.~Grollmann and A.~L. Selman.
\newblock Complexity measures for public-key cryptosystems.
\newblock \emph{SIAM J. Comput.}, 17(2):309--335, 1988.

\bibitem[GSV98]{GoldreichSV98}
O.~Goldreich, A.~Sahai, and S.~P. Vadhan.
\newblock Honest-verifier statistical zero-knowledge equals general statistical
  zero-knowledge.
\newblock In \emph{STOC}, pages 399--408. 1998.

\bibitem[IOS97]{ItohOS97}
T.~Itoh, Y.~Ohta, and H.~Shizuya.
\newblock A language-dependent cryptographic primitive.
\newblock \emph{J. Cryptology}, 10(1):37--50, 1997.

\bibitem[LLL82]{lenstra82:_factor}
A.~K. Lenstra, H.~W. Lenstra, Jr., and L.~Lov\'{a}sz.
\newblock Factoring polynomials with rational coefficients.
\newblock \emph{Mathematische Annalen}, 261(4):515--534, December 1982.

\bibitem[MR04]{DBLP:journals/siamcomp/MicciancioR07}
D.~Micciancio and O.~Regev.
\newblock Worst-case to average-case reductions based on {Gaussian} measures.
\newblock \emph{SIAM J. Comput.}, 37(1):267--302, 2007.
\newblock Preliminary version in FOCS 2004.

\bibitem[MV03]{DBLP:conf/crypto/MicciancioV03}
D.~Micciancio and S.~P. Vadhan.
\newblock Statistical zero-knowledge proofs with efficient provers: Lattice
  problems and more.
\newblock In \emph{CRYPTO}, pages 282--298. 2003.

\bibitem[MV10]{DBLP:conf/stoc/MicciancioV10}
D.~Micciancio and P.~Voulgaris.
\newblock A deterministic single exponential time algorithm for most lattice
  problems based on {Voronoi} cell computations.
\newblock In \emph{STOC}, pages 351--358. 2010.

\bibitem[Pei07]{peikert08:_limit_hardn_of_lattic_probl}
C.~Peikert.
\newblock Limits on the hardness of lattice problems in $\ell_p$ norms.
\newblock \emph{Computational Complexity}, 17(2):300--351, May 2008.
\newblock Preliminary version in CCC 2007.

\bibitem[Pei09]{DBLP:conf/stoc/Peikert09}
C.~Peikert.
\newblock Public-key cryptosystems from the worst-case shortest vector problem.
\newblock In \emph{STOC}, pages 333--342. 2009.

\bibitem[Pei10]{DBLP:conf/crypto/Peikert10}
C.~Peikert.
\newblock An efficient and parallel {Gaussian} sampler for lattices.
\newblock In \emph{CRYPTO}, pages 80--97. 2010.

\bibitem[PW08]{DBLP:conf/stoc/PeikertW08}
C.~Peikert and B.~Waters.
\newblock Lossy trapdoor functions and their applications.
\newblock In \emph{STOC}, pages 187--196. 2008.

\bibitem[Reg03]{DBLP:journals/jacm/Regev04}
O.~Regev.
\newblock New lattice-based cryptographic constructions.
\newblock \emph{J. ACM}, 51(6):899--942, 2004.
\newblock Preliminary version in STOC 2003.

\bibitem[Reg05]{DBLP:journals/jacm/Regev09}
O.~Regev.
\newblock On lattices, learning with errors, random linear codes, and
  cryptography.
\newblock \emph{J. ACM}, 56(6):1--40, 2009.
\newblock Preliminary version in STOC 2005.

\bibitem[SV03]{SahaiV03}
A.~Sahai and S.~P. Vadhan.
\newblock A complete problem for statistical zero knowledge.
\newblock \emph{J. ACM}, 50(2):196--249, 2003.

\end{thebibliography}

\end{document}